\documentclass[11pt]{article}
\usepackage[utf8]{inputenc}
\usepackage[T1]{fontenc}
\usepackage{lmodern}
\usepackage{amsmath,amssymb,amsfonts,amsthm,mathtools,bm}
\usepackage{bbm}
\AtBeginDocument{%
  \begingroup\usefont{U}{bbm}{m}{n}\endgroup
  \DeclareFontShape{U}{bbm}{m}{n}{%
    <5> bbm5 <6> bbm6 <7> bbm7 <8> bbm8 <9> bbm9 <10> bbm10 <12> bbm12
    <5.5> bbm5 <10.95> bbm10 <14.4> bbm12
    <17.28><20.74><24.88> bbm17}{}}
\usepackage{enumitem}
\usepackage{array}
\usepackage{booktabs}
\usepackage{longtable}
\usepackage{enumitem}
\usepackage{booktabs}
\usepackage{algorithm}
\usepackage{algpseudocode}
\usepackage{natbib}
\usepackage[colorlinks=true,linkcolor=blue,citecolor=blue,urlcolor=blue]{hyperref}
\usepackage{xcolor}
\usepackage{graphicx}
\usepackage{float}
\graphicspath{{./}}
\usepackage{adjustbox}
\usepackage{multirow}
\usepackage[margin=1.0in]{geometry}
\newtheorem{theorem}{Theorem}
\newtheorem{lemma}{Lemma}
\newtheorem{proposition}{Proposition}
\newtheorem{corollary}{Corollary}
\newtheorem{assumption}{Assumption}

\newtheorem{remark}{Remark}
\newtheorem{example}{Example}

\newcommand{\E}{\mathbb E}

\newcommand{\ind}[1]{\mathbbm 1\{#1\}}

\newcommand{\Cov}{\mathbb C\mathrm{ov}}

\newcommand{\calK}{\mathcal K}

\newcommand{\calU}{\mathcal U}
\newcommand{\calX}{\mathcal X}

\newcommand{\Prob}{\mathbb{P}}
\newcommand{\Pn}{\mathbb{P}_n}

\newcommand{\one}{\mathbbm{1}}
\newcommand{\indep}{\perp\!\!\!\perp}
\newcommand{\X}{\bm{X}}
\newcommand{\calO}{\mathcal{O}}
\newcommand{\zbar}{\bar{z}}

\usepackage{setspace}
\usepackage{bibspacing}
\newcommand{\blind}{0}

\begin{document}

\def\spacingset#1{\renewcommand{\baselinestretch}%
{#1}\small\normalsize} \spacingset{1}


\if0\blind
{
  \title{Sharp bounds for population and principal effects under latent sequential ignorability} 
  \author{Yuki Ohnishi\\
    Department of Statistics,  
    University of Florida\\
    yuki.ohnishi@ufl.edu}
  \maketitle
} \fi

\if1\blind
{
  \bigskip
  \bigskip
  \bigskip
  \begin{center}
    {\LARGE\bf }
\end{center}
  \medskip
} \fi


\begin{abstract}
Treatment decisions are often sequential and depend on information absent from the recorded history. Latent sequential ignorability models such assignment by assuming that baseline covariates and principal strata, defined by joint potential intermediate responses, account for confounding of later treatment. For two treatment periods with a binary intermediate response, we derive sharp population and principal effect bounds under a specified association between the potential responses and show that the two targets can require different optimizations to obtain sharp bounds. For population effects, optimization over latent treatment probabilities reduces to their feasible endpoints. Stronger selection between response types can be offset by making their outcome distributions more similar, preserving both the observed distribution and population marginal outcome distributions. This replacement can change effects within principal strata, whose sharp bounds may instead require interior treatment probabilities. Restrictions on differences in conditional effects can tighten the bounds without equating outcome levels. For binary outcomes, we give an exact finite algorithm that includes these interior assignments. We also develop inference under smoothness and estimation conditions. We illustrate the bounds in a study of gym attendance after an incentive program.
\end{abstract}

\noindent%

\spacingset{1.1} 

\section{Introduction}\label{sec:intro}\label{sec:contributions}\label{sec:related}

Treatment is often assigned sequentially, with later decisions guided by responses to earlier interventions and information that may be absent from the recorded history. For example, a physician's decision to continue medication may depend on adverse effects experienced by the patient but not captured in the study records, while a follow-up marketing offer may depend on private assessments of a customer's purchase intentions. Such information can also predict later health or purchasing outcomes, so adjustment for recorded history alone need not remove confounding \citep{robins1986,robins2000}. 

We ask what can be learned about causal effects when the hidden information has a scientifically interpretable structure. Specifically, latent sequential ignorability (LSI) assumes that baseline covariates, initial treatment and a latent response profile suffice to account for confounding of later treatment \citep{ricciardi2020}. The profile consists of the potential intermediate responses under the two initial treatments and defines a principal stratum \citep{frangakis2002}. Initial treatment remains exchangeable given baseline covariates. The data then identify each response margin, but generally not the association between potential responses. Even with that association specified, unknown treatment probabilities and outcome distributions within strata can prevent point identification of population and principal effects. \citet{ricciardi2020} develop Bayesian inference using additional parametric models. We examine what LSI and the observed data imply without parametric component outcome models.

We study two binary treatment periods with a binary intermediate response. For a specified response coupling, Section~\ref{sec:identification} characterizes the latent outcome means compatible with the full observed outcome distributions and one common later-treatment mechanism. This yields sharp population and principal effect bounds while allowing latent outcome means to differ. Our main result is a contraction that moves latent assignment to a feasible endpoint while preserving the observed law and population marginal outcome distributions. Population bounds therefore require only two assignment choices within each response cell and covariate value, with outcome distributions still optimized. Outcome distributions within principal strata can change, and sharp principal effect bounds may require interior assignments, even when the optimizing assignment is unique.

Restrictions on effect heterogeneity between response types can further sharpen these bounds without equating outcome levels, while preserving the population reduction. For binary outcomes, we give an exact finite search that includes the interior assignments needed for principal bounds. Section~\ref{sec:estimation} develops one-step inference under stated smoothness, stability and estimation conditions for prespecified finite sensitivity grids. The supplement gives a finite search over unspecified response couplings for binary population contrasts without effect restrictions (Section~\ref{supp:unspecified-computation}) and characterizes outcome compatibility when selection remains within a response profile (Section~\ref{app:residual}). Simulations assess estimation and coverage, and a gym attendance study illustrates the bounds.

Our component bounds use established trimming arguments \citep{lee2009,dornguo2023,dornguokallus2025}. The response coupling follows \citet{tong2025}, and principal weighting builds on \citet{ding2017} and \citet{jiang2022}. Longitudinal sensitivity bounds \citep{tan2025} and models combining assignment and outcome restrictions \citep{zhangxutan2026,zhangxutandist2025} address related problems. Here response margins and coupling fix latent prevalence, both outcome arms share one later-treatment mechanism, and restrictions link their conditional effects. Section~\ref{supp:matched-comparisons} and Table~\ref{tab:contribution-comparison} give matched comparisons. Complementary approaches study longitudinal treatment noncompliance \citep{ohnishi2025}, dynamic regimes \citep{park2025}, negative controls \citep{luo2024} and heterogeneous latent means \citep{nguyen2024,wuantonelli2024}. Other work studies longitudinal bias sensitivity \citep{robinsrotnitzky2000,brumback2004,yanglok2018} and distributional sensitivity for a single treatment \citep{mastenpoirier2025}.

\section{Sequential decisions and response profiles}\label{sec:setup}
\subsection{Data, interventions and estimands}
We observe independent copies of $\calO=(\X,Z_1,S,Z_2,Y)$, with baseline covariates $\X$, binary treatments $Z_1,Z_2$, a binary response $S$ measured between them, and a final outcome $Y$. Potential responses are $S(j)$ and potential outcomes are $Y(j,a)$, where $j,a\in\{0,1\}$ index interventions on $Z_1$ and $Z_2$, respectively. The gym study in Section~\ref{sec:application} examines whether attendance persists after an incentive program. There, $Z_1$ is intended assignment to the program or placebo, $S$ records week-1 attendance, $Z_2$ records week-2 attendance, and $Y$ records week-5 attendance.

The intermediate response profile $U=(S(0),S(1))\in\calU=\{0,1\}^2$ defines the principal strata \citep{frangakis2002}. For example, $U=(1,1)$ comprises members who would attend in week 1 under either assignment. These latent types need not be observed by the decision-maker. We impose no monotonicity.

The estimands of interest are the population sequence mean and the principal sequence mean:
\begin{equation}\label{eq:targets28}
 \psi(j,a)=\E\{Y(j,a)\},\qquad
 \psi_u(j,a)=\E\{Y(j,a)\mid U=u\}.
\end{equation}
Their differences answer distinct questions. The contrast $\psi(j,1)-\psi(j,0)$ concerns the population under initial assignment $j$. The principal effect $\tau_u=\psi_u(j,1)-\psi_u(j,0)$ concerns a response type defined before the later treatment. All interventions must specify their treatment versions, and the analysis population must be held fixed across comparisons.

\subsection{Assignment assumptions}\label{sec:assumptions-assignment}
Our analysis uses the following assumptions.
\begin{assumption}[Consistency]\label{ass:sutva}
The recorded response and outcome satisfy $S=S(Z_1)$ and $Y=Y(Z_1,Z_2)$, with well-defined treatment versions and no interference.
\end{assumption}
\begin{assumption}[Latent sequential ignorability]\label{ass:lsi}
Let $\mathcal F=\{U,Y(j,a):(j,a)\in\{0,1\}^2\}$ collect the potential responses and outcomes.
\begin{enumerate}[label=(\alph*),nosep]
\item $Z_1\indep\mathcal F\mid\X$;
\item $Z_2\indep\{Y(j,a):(j,a)\in\{0,1\}^2\}\mid\X,Z_1,U$.
\end{enumerate}
\end{assumption}
Part (a) concerns the initial treatment. Part (b) permits later treatment to depend on the response profile, but assumes no remaining treatment selection within it. This is the assignment model of \citet{ricciardi2020}, abbreviated LSI. Ordinary sequential ignorability (SI) instead requires
\begin{equation}\label{eq:si}
 Z_1\indep\mathcal F\mid\X,\qquad
 Z_2\indep\{Y(j,a):(j,a)\in\{0,1\}^2\}\mid\X,Z_1,S.
 \tag{SI}
\end{equation} Neither condition generally implies the other, because adding an unobserved potential response can change conditional independence. Section~\ref{app:relations} gives the implications and counterexamples. Section~\ref{sec:residual-selection} examines departures from part (b).

\begin{assumption}[Positivity]\label{ass:positivity}
For some $\epsilon>0$, $\epsilon\le\Prob(Z_1=j\mid\X)\le1-\epsilon$ and $\epsilon\le\Prob(Z_2=1\mid\X,Z_1=j,U=u)\le1-\epsilon$ whenever $e_u(\X)=\Prob(U=u\mid\X)>0$.
\end{assumption}
The latent positivity margin $\epsilon$ restricts the assignment law and is distinct from numerical truncation of fitted probabilities.

\subsection{When a response profile accounts for selection}\label{sec:response-sufficiency}
A sufficient structural mechanism illustrates LSI. Suppose a latent disposition $L$ determines the potential responses through $S(j)=s_j(\X,L)$. Later treatment and outcomes satisfy $Z_2=\ind{E_A\le h_{Z_1,U}(\X)}$ and $Y(j,a)=g_{ja}(\X,L,E_Y)$. Here $s_j$ and $g_{ja}$ are general measurable response and outcome functions, and $h_{j,u}$ is a measurable treatment probability function taking values in $[\epsilon,1-\epsilon]$. The variable $E_A$ is uniform on $(0,1)$ and independent of $(L,E_Y)$ given $\X$, and $Z_1$ is independent of $(L,E_Y,E_A)$ given $\X$. This construction satisfies LSI because assignment depends on $L$ only through $U$, even when outcomes depend on variation in $L$ within $U$.

An unmeasured trait $H$ that varies within $(\X,U)$ and affects both $h_{j,U}(\X,H)$ and $g_{ja}(\X,L,H,E_Y)$ can confound later treatment within a stratum. In the gym study, members with the same covariates and early-response profile may differ in motivation or work schedules that affect both week-2 and week-5 attendance. Initial randomization does not remove this residual confounding, so response sufficiency remains a substantive restriction.

Sensitivity parameters distinguish assumptions that the observed data cannot resolve. The coupling parameter $\theta$ in Assumption~\ref{ass:theta} specifies the association between $S(0)$ and $S(1)$, whose joint distribution is unobserved. The margin $\epsilon$ restricts latent treatment probabilities. The effect budget $\gamma$, introduced in Section~\ref{sec:effect-restrictions}, bounds differences in conditional mean effects between response types under LSI. Separately, the allowance $\eta$ in Supplementary Section~\ref{app:residual} bounds conditional mean departures from LSI within a profile, which an effect restriction does not address.

\subsection{Observed information and response coupling}\label{sec:nuisance}\label{sec:coupling}
We now define the observed conditional probabilities and outcome distributions, together with their latent counterparts. Write $\pi_j(\X)=\Prob(Z_1=j\mid\X)$, $p_j(\X)=\Prob(S=1\mid Z_1=j,\X)$ and $p_j^s(\X)=p_j(\X)^s\{1-p_j(\X)\}^{1-s}$.
For a profile $u=(u_0,u_1)$, $u_j$ is its response under initial treatment $j$. The profiles compatible with the recorded cell $(j,s)$ form $\calU(j,s)=\{u:u_j=s\}$. Their within-cell probabilities and latent treatment probabilities are
\begin{equation}\label{eq:weights}
\begin{aligned}
 w_u(j,s,\X)&=\Prob(U=u\mid\X,Z_1=j,S=s)=\frac{e_u(\X)}{p_j^s(\X)},\\
 h_u(j,\X)&=\Prob(Z_2=1\mid\X,Z_1=j,U=u).
\end{aligned}
\end{equation}
The ratio follows from consistency and initial treatment ignorability. Averaging these latent treatment probabilities over the compatible profiles gives the observed later-treatment probability,
\[
q_1(j,s,\X)=\Prob(Z_2=1\mid\X,Z_1=j,S=s)
=\sum_{u\in\calU(j,s)}w_u(j,s,\X)h_u(j,\X),
\]
with $q_0=1-q_1$. Let $F_{ja,s,\X}$ denote the conditional distribution of $Y$ given $(\X,Z_1=j,S=s,Z_2=a)$, with mean $\mu_{(j,a)}(s,\X)=\E(Y\mid\X,Z_1=j,S=s,Z_2=a)$. The latent mean is $m_{(j,a)}(u,\X)=\E\{Y(j,a)\mid U=u,\X\}$.

In cell $(j,s)$, only the alternative response $V=S(1-j)$ is unobserved. For $v\in\{0,1\}$, let $u^{(v)}$ be the profile with response $s$ under treatment $j$ and $v$ under treatment $1-j$. Then $w=w_{u^{(1)}}(j,s,\X)=\Prob\{V=1\mid S(j)=s,\X\}$ is the proportion with alternative response one before later-treatment selection. Assumption~\ref{ass:lsi}(a) identifies the two response margins, but not the joint distribution of $S(0)$ and $S(1)$. Specifying their association supplies the missing information needed to determine the principal scores $e_u$ and weights $w_u$.
The response margins imply the Fr\'echet bounds
\[
 \max\{0,p_0(\X)+p_1(\X)-1\}\le e_{11}(\X)\le\min\{p_0(\X),p_1(\X)\},
\]
which define the feasible coupling interval at each covariate value.
We first derive bounds conditional on a specified response coupling, treating its association parameter as a sensitivity input. Without this specification, the identified set is the union over all couplings compatible with the response margins.
\begin{assumption}[Specified response coupling]\label{ass:theta}
The margins lie in $(0,1)$ and the conditional odds ratio
\[
 \theta(\X)=\frac{e_{11}(\X)e_{00}(\X)}{e_{10}(\X)e_{01}(\X)}
\]
is an analyst-specified function taking values in $[0,\infty]$.
\end{assumption}
This margin-free parametrization follows \citet{tong2025}. Sensitivity in $\theta$ varies the unobserved association between the potential responses, separately from later-treatment selection. The value $\theta=1$ gives conditional independence of $S(0)$ and $S(1)$ given $\X$. The upper coupling endpoint, attained at $\theta=\infty$, gives the strongest positive association allowed by the margins and implies monotonicity $S(1)\ge S(0)$ when $p_1\ge p_0$ almost surely. The lower coupling endpoint, attained at $\theta=0$, gives the strongest negative association allowed by the margins. Varying $\theta$ therefore assesses dependence assumptions while preserving the identified margins.

The response margins and specified odds ratio uniquely determine the principal scores $e_u$. Lemma~\ref{lem:plackett} in Supplementary Section~\ref{supp:plackett} records this result and its smoothness properties, with proof in Section~\ref{supp:proof-plackett}.

\paragraph{Point-identification benchmark.}
The following additional restriction provides a point-identification benchmark. It is not imposed in our sharp-bound analysis.
\begin{assumption}[Sequential principal ignorability]\label{ass:spi}
For every $(j,a)\in\{0,1\}^2$ and $s\in\{0,1\}$, $ m_{(j,a)}(u,\X)=m_{(j,a)}(u',\X)$
almost surely for $u,u'\in\calU(j,s)$ with $e_u(\X)e_{u'}(\X)>0$.
\end{assumption}
\noindent
Conditional on baseline covariates and $S(j)$, sequential principal ignorability (SPI) excludes additional information about the mean of $Y(j,a)$ from the alternative response $S(1-j)$, although that response may still predict later treatment. Only conditional means must agree. LSI instead restricts later-treatment selection within a response profile. Our bound analysis retains LSI while allowing SPI to fail.

\section{Compatible population and principal effects}\label{sec:identification}\label{sec:bounds}

\subsection{One representation for population and principal effects}\label{sec:unified-identification}
We introduce a common mixture representation for population and principal effects that separates treatment selection from effect differences between response types. Throughout, the response coupling, covariate value and cell $(Z_1,S)=(j,s)$ are fixed. The alternative response $V=S(1-j)$ has within-cell probability $w=\Prob\{V=1\mid S(j)=s,\X\}$. We consider mixed cells, where $0<w<1$, and suppress covariate conditioning when unambiguous.

The latent treatment probability for type $v\in\{0,1\}$ is $h_v=\Prob(Z_2=1\mid Z_1=j,S=s,V=v)$, and the observed probability is $q=q_1(j,s,\X)=wh_1+(1-w)h_0$. With $h=h_1$, the treatment probability for type $V=0$ is therefore $h_0=(q-wh)/(1-w)$. For $a\in\{0,1\}$, $F_a$ denotes the observed distribution of $Y$ given $(Z_1=j,S=s,Z_2=a)$ at the fixed covariate value. Under LSI, $F_a$ mixes the two response types' potential-outcome distributions. Its component share $t_a(h)=\Prob(V=1\mid\X,Z_1=j,S=s,Z_2=a)$ is the proportion of type $V=1$ in treatment group $a$. Bayes' rule gives
\[
t_1(h)=\frac{wh}{q},\qquad t_0(h)=\frac{w(1-h)}{1-q}.
\]
The observed mean is $\mu_a=\int y\,dF_a(y)$, and the latent mean for type $v$ is $m_a(v)=\E\{Y(j,a)\mid V=v,S(j)=s,\X\}$, with gap $\delta_a=m_a(1)-m_a(0)$. Under LSI, $\mu_a=t_am_a(1)+(1-t_a)m_a(0)$ uses the selected share $t_a$, whereas the population cell mean $M_a=wm_a(1)+(1-w)m_a(0)$ uses the share $w$ before selection. These means satisfy
\begin{equation}\label{eq:unified-means}
M_a=\mu_a-(t_a-w)\delta_a,\qquad
m_a(v)=M_a+(v-w)\delta_a.
\end{equation}
In the first equality, $t_a-w$ measures the change in response-type composition due to treatment selection. Its product with $\delta_a$ is the difference between the observed and population cell means. The second equality recovers each type's conditional mean by adding $(1-w)\delta_a$ for type $V=1$ or subtracting $w\delta_a$ for type $V=0$. The corresponding effect decomposition is
\begin{equation}\label{eq:unified-effect}
\begin{split}
m_1(v)-m_0(v)
={}\mu_1-\mu_0-\{(t_1-w)\delta_1-(t_0-w)\delta_0\}
+(v-w)(\delta_1-\delta_0).
\end{split}
\end{equation}
The bracketed term in \eqref{eq:unified-effect} is the difference between the observed contrast and the population cell effect. The final term is the difference between a principal effect and that population effect. Latent selection occurs when treatment probabilities differ between response types within the same observed cell. Figure~\ref{fig:mechanism} illustrates how treatment selection and effect heterogeneity can make the observed contrast, population cell effect and principal effects differ. In panel (A), selection increases the proportion of type $V=1$ from $40\%$ in the cell to $60\%$ among the treated and reduces it to $20\%$ among the untreated. Panel (B) shows that correcting this selection reduces the observed contrast from $0.42$ to the population effect $0.32$, while effect heterogeneity gives principal effects $0.20$ and $0.50$. Averaging over $v$ removes only the final term in \eqref{eq:unified-effect}. Section~\ref{sec:bias} of the Supplementary Material gives the mixture identities and the resulting g-formula bias.

Homogeneous treatment assignment across response types, i.e., $h_0=h_1=q$, gives $t_1=t_0=w$ and identifies the population contrast, although principal effects may still differ when $\delta_1\ne\delta_0$. Conversely, $\delta_1=\delta_0$ equates conditional principal effects while permitting selection on unequal outcome levels.

\begin{figure}[t]
\centering
\includegraphics[width=\textwidth]{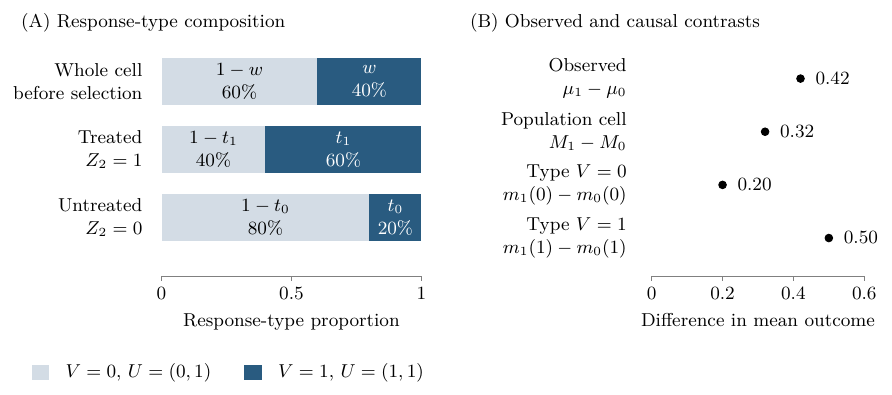}
\caption{Selection and effect heterogeneity in the cell $(Z_1,S)=(1,1)$ at a fixed covariate value.
Panel A shows the response-type proportions in the whole cell and among those receiving $Z_2=1$ or $Z_2=0$. The dark segments give the type $V=1$ shares $w$ before selection and $t_a$ within treatment group $a$; the light segments give their complements. Panel B compares the observed mean difference $\mu_1-\mu_0$, population cell effect $M_1-M_0$ and principal effects $m_1(v)-m_0(v)$ of later treatment $a=1$ versus $a=0$.
Here $V=S(0)$, $w=0.4$, $h_1=0.75$, $h_0=1/3$ and $q=0.5$. The component outcome means under $(Z_2=0,Z_2=1)$ are $(0.2,0.4)$ for $V=0$ and $(0.3,0.8)$ for $V=1$. This constructed Bernoulli law illustrates the selection and effect-heterogeneity adjustments in \eqref{eq:unified-effect}.}
\label{fig:mechanism}
\end{figure}

These identities do not make every choice of $h$ and $\delta_a$ feasible. The component outcome distributions must reproduce each observed $F_a$, and both latent treatment probabilities must satisfy positivity. For the outcome restriction, a distribution $F$ with mean $\mu$ and quantile function $F^{-1}$ has lower and upper tail integrals
\[
L_F(t)=\int_0^tF^{-1}(r)\,dr,\qquad
U_F(t)=\int_{1-t}^1F^{-1}(r)\,dr.
\]
For $0<t<1$, the smallest and largest possible means of a component comprising share $t$ of $F$ are $L_F(t)/t$ and $U_F(t)/t$. They allocate the lowest or highest fraction $t$ of the outcomes to that component, splitting mass at the cutoff when necessary. Since its mean is $\mu+(1-t)\delta$, the compatible component-mean gaps form the interval
\begin{equation}\label{eq:unified-delta-interval}
\mathcal D_F(t)=
\left[\frac{L_F(t)/t-\mu}{1-t},\,
\frac{U_F(t)/t-\mu}{1-t}\right]
=:[\ell_F(t),u_F(t)].
\end{equation}
For assignment compatibility, both $h$ and $h_0=(q-wh)/(1-w)$ must lie in $[\epsilon,1-\epsilon]$. These inequalities give the feasible interval
\begin{equation}\label{eq:unified-h-interval}
\mathcal H_\epsilon(w,q)=
\left[
\max\left\{\epsilon,\frac{q-(1-\epsilon)(1-w)}w\right\},\,
\min\left\{1-\epsilon,\frac{q-\epsilon(1-w)}w\right\}
\right].
\end{equation}
Thus $\mathcal H_\epsilon$ restricts assignment, while $\mathcal D_{F_a}\{t_a(h)\}$ restricts outcome-mean differences at the induced component shares. Both outcome arms must use the same $h$ because they arise from one treatment mechanism. The following theorem shows that these restrictions characterize all compatible latent means and yield sharp bounds for scalar linear targets, including population and principal contrasts.

\begin{theorem}[Compatible means and sharp causal contrasts]\label{thm:unified-contrast}
Suppose consistency and first-stage ignorability hold and $\calX$ is a standard Borel space. Fix a measurable response coupling reproducing the observed margins, a bounded outcome support $[y_-,y_+]$, and $\epsilon>0$ such that all observed treatment probabilities and their complements are at least $\epsilon$. Consider full-data laws inducing the observed distribution, satisfying second-stage LSI, and having latent treatment probabilities in $[\epsilon,1-\epsilon]$. All conditional statements concern cells with $p_j^s(\X)>0$; zero-probability cells are omitted.

The compatible latent mean functions are exactly \eqref{eq:unified-means}, with a measurable choice $h\in\mathcal H_\epsilon(w,q)$ in every mixed cell and separate choices $\delta_a\in\mathcal D_{F_a}\{t_a(h)\}$ for its two outcome arms. In a pure cell, where only one response type has positive probability, its sole latent mean is $\mu_a$.

For bounded measurable coefficients $c_{j,a,u}(\X)$, define
\[
T_c=\E\sum_{j,a,u}e_u(\X)c_{j,a,u}(\X)m_{(j,a)}(u,\X).
\]
In cell $(j,s)$ let $c_{av}=c_{j,a,u^{(v)}}(\X)$ and put
\[
A_a=wc_{a1}+(1-w)c_{a0},\qquad
B_a(h)=wc_{a1}\{1-t_a(h)\}-(1-w)c_{a0}t_a(h).
\]
Its lower and upper contributions are
\begin{equation}\label{eq:unified-contrast-bounds}
\begin{split}
K^-_{j,s}(\X)
&=\min_{h\in\mathcal H_\epsilon}
\sum_a\left[A_a\mu_a+
\min_{d\in\{\ell_{F_a}(t_a),u_{F_a}(t_a)\}}B_a(h)d\right],\\
K^+_{j,s}(\X)
&=\max_{h\in\mathcal H_\epsilon}
\sum_a\left[A_a\mu_a+
\max_{d\in\{\ell_{F_a}(t_a),u_{F_a}(t_a)\}}B_a(h)d\right].
\end{split}
\end{equation}
For pure cells use $K^-_{j,s}=K^+_{j,s}=\sum_a c_{av}\mu_a$ for their sole $v$. The sharp identified set for $T_c$ is
\[
\left[\E\sum_{j,s}p_j^s(\X)K^-_{j,s}(\X),\quad
\E\sum_{j,s}p_j^s(\X)K^+_{j,s}(\X)\right].
\]
Every value is attained by a compatible full law.
\end{theorem}

At each feasible assignment, the target is linear in the two outcome-mean gaps, so its extrema occur at their interval endpoints. Optimization over the common assignment and averaging across cells then give the sharp scalar bounds. Every value between them is attainable under the full observed outcome distributions and the maintained coupling, LSI and positivity margin. Separately optimized scalar extrema need not be jointly attainable. The proof in Section~\ref{supp:proof-unified-contrast} constructs compatible component distributions and a full-data law.

Population sequence contrasts use coefficients constant across $u$. A principal contrast in stratum $u_*$ instead uses coefficients zero outside that stratum and divides both bounds by $\E(e_{u_*})>0$. This representation covers both targets and the point-identification benchmarks in Section~\ref{sec:spi}. With unspecified coupling, the identified set is the union over measurable couplings reproducing the margins. For principal effects, the stratum mass is recomputed at each coupling.

Matching the full observed outcome distribution can sharpen bounds beyond those based on its support and mean alone. For example, in a cell with $w=q=1/2$, $\epsilon=1/4$ and $F_0=F_1=\operatorname{Unif}[0,1]$, the sharp interval for each population cell mean $M_a$ is $[3/8,5/8]$, compared with $[1/3,2/3]$ when only the support and mean are retained. The wider interval allows component means that cannot reproduce the uniform observed law (Supplementary Section~\ref{supp:sens-est}).

\subsection{The population-law contraction and its principal limitation}\label{sec:contraction}
The following theorem is the structural reduction underlying the bound calculations. Write $G_{av}$ for the law of $Y(j,a)$ in type $v$ and $\overline G_a=wG_{a1}+(1-w)G_{a0}$ for its population cell law.

\begin{theorem}[Law-preserving contraction]\label{thm:contraction}
Under Theorem~\ref{thm:unified-contrast}, impose separately at each mixed cell and covariate value a restriction $(\delta_0,\delta_1)\in\mathcal C_{j,s}(\X)$ with measurable graph and closed sections, independent of $h$, containing zero, and satisfying
$\alpha\mathcal C_{j,s}(\X)\subseteq\mathcal C_{j,s}(\X)$ for $0\le\alpha\le1$.
Every compatible collection of population cell marginal laws $(\overline G_0,\overline G_1)$ admits a compatible full-data completion with assignment at an endpoint of $\mathcal H_\epsilon(w,q)$ in every mixed cell. Consequently, restricting assignment to these endpoints leaves the attainable collection of population cell marginal laws unchanged. Population sequence and shared-assignment contrast bounds require only these two assignment choices, with the outcome laws still optimized subject to the restriction.
\end{theorem}
\begin{proof}
A singleton assignment interval requires no change. Otherwise, for $h\ne q$ choose its endpoint $h^*$ on the same side of $q$ as $h$. For $h=q$, choose the upper endpoint if it exceeds $q$, and the lower endpoint otherwise. Put
\[
\alpha=\frac{h-q}{h^*-q}\in[0,1],\qquad
G^*_{av}=\alpha G_{av}+(1-\alpha)\overline G_a,
\]
with $\alpha=0$ when $h=q$. The same $\alpha$ applies to both outcome arms. Since $t_a(h)-w=\alpha\{t_a(h^*)-w\}$,
\begin{equation}\label{eq:contraction-laws}
\begin{aligned}
t_a(h^*)G^*_{a1}+\{1-t_a(h^*)\}G^*_{a0}=F_a, \quad
wG^*_{a1}+(1-w)G^*_{a0}=\overline G_a.
\end{aligned}
\end{equation}
Thus both observed and population cell laws are preserved and the new gap vector is $\alpha(\delta_0,\delta_1)\in\mathcal C_{j,s}(\X)$. The endpoint rule and transformed kernels are measurable. The completion in Theorem~\ref{thm:unified-contrast} joins them with the fixed response coupling and one common assignment per cell. The reverse inclusion is immediate because endpoint assignments are feasible assignments.
\end{proof}

Under the stated restrictions and a fixed response coupling, Theorem~\ref{thm:contraction} replaces any compatible model by one using an assignment endpoint in each mixed cell while preserving the observed distribution and the population marginal law of each $Y(j,a)$. This does not assert that the true assignment probabilities lie at endpoints, but it says an endpoint model can reproduce the same population causal quantities. Population mean and contrast bounds therefore require only the two endpoint assignments per cell, with outcome distributions still optimized. The restriction need not be convex. It must be preserved when both latent outcome-mean differences are shrunk toward zero by the same factor.

Principal laws $G_{av}$ generally change in this construction. Their extrema can therefore require an interior assignment even when population extrema do not. 

\begin{example}[Failure of assignment-endpoint reduction for principal effects]\label{ex:principal-endpoint-failure}
Suppose $F_0=F_1=\operatorname{Bernoulli}(1/2)$, $w=q=1/2$ and $0<\epsilon<1/2$. For the stratum with $V=1$, the conditional upper effect is
\[
\max_{\epsilon\le h\le1-\epsilon}
\left[\min\{1,(2h)^{-1}\}
-\max\{0,1-[2(1-h)]^{-1}\}\right]=1.
\]
It is attained at $h=1/2$, whereas either assignment endpoint gives $1/\{2(1-\epsilon)\}<1$. The lower bound is $-1$. A full law attaining the upper bound has $Y(j,1)=v$, $Y(j,0)=1-v$ and assignment independent of $v$. Its population effect is zero, yet its principal effects are $1$ and $-1$. Population and principal uncertainty therefore require different optimizations even within the same compatible model.
\end{example}

Theorem~\ref{thm:contraction} preserves whole population marginal outcome distributions, so its implications extend beyond means. For example, the population distribution function of $Y(j,a)$ is
\[
H_{ja}(y)=\E\sum_s p_j^s(\X)\overline G_{a,j,s,\X}(( -\infty,y])
\]
and has the same attainable values under arbitrary feasible assignments and endpoint assignments, simultaneously across thresholds and treatment sequences. The stated outcome-mean restrictions remain in force. Therefore, the reduction to assignment endpoints also applies to nonlinear functionals of these population marginal distributions, such as quantiles, variances and differences of medians. Their bounds require optimizing the target after combining the cell laws into population distributions. For example, averaging cell medians need not give the population median.

\subsection{Population and principal bound computations}
We now derive sharp population contrast bounds and, for binary outcomes, give an exact finite assignment search for principal contrast bounds.
We denote the endpoints of \eqref{eq:unified-h-interval} by $h^-_\epsilon,h^+_\epsilon$. Equation~\eqref{eq:unified-means} gives $M_a-\mu_a=(w-t_a)\delta_a$. Linearity in $\delta_a$ gives the largest upward and downward deviations over its compatible interval,
\begin{equation}\label{eq:trimmed-deviations}
 E_a^+(t)=\max_{d\in\{\ell_{F_a}(t),u_{F_a}(t)\}}(w-t)d,\qquad
 E_a^-(t)=-\min_{d\in\{\ell_{F_a}(t),u_{F_a}(t)\}}(w-t)d.
\end{equation}
The compatible gap interval in Theorem~\ref{thm:unified-contrast} uses the full observed outcome distribution. At a fixed feasible assignment $h$, the attainable population cell means therefore satisfy
\[
M_a\in\left[\mu_a-E_a^-\{t_a(h)\},\quad
                 \mu_a+E_a^+\{t_a(h)\}\right].
\]
We write $E_F^\pm$ when only one outcome law is under consideration. Because the two arm-specific tilts $t_1(h)$ and $t_0(h)$ move away from $w$ in opposite directions as $h$ moves away from $q$, the population contrast has a simple solution.
\begin{proposition}[Sharp contrast bounds with shared assignment]\label{prop:shared-tilt}
Under Theorem~\ref{thm:unified-contrast}, define in each mixed cell
\begin{align}
J^+_\epsilon=\max_{h\in\{h^-_\epsilon,h^+_\epsilon\}}
\{E_1^+(t_1(h))+E_0^-(t_0(h))\}, \quad
J^-_\epsilon=\max_{h\in\{h^-_\epsilon,h^+_\epsilon\}}
\{E_1^-(t_1(h))+E_0^+(t_0(h))\},\label{eq:shared-tilt-formula}
\end{align}
and set both to zero in pure cells. With $\Delta_g=\E\sum_s p^s_{z_1}\{\mu_{(z_1,1)}(s,\X)-\mu_{(z_1,0)}(s,\X)\}$, the identified set for $\psi(z_1,1)-\psi(z_1,0)$ is
\begin{equation}\label{eq:shared-tilt-region}
\left[\Delta_g-\E\sum_s p^s_{z_1}J^-_\epsilon,\quad
\Delta_g+\E\sum_s p^s_{z_1}J^+_\epsilon\right].
\end{equation}
Only two evaluations per cell are required. At a common fixed coupling, separate endpoint differencing is sharp for sequences with different first-period arms, whose latent assignment mechanisms are unrestricted relative to each other.
\end{proposition}
The proof is in Section~\ref{supp:proof-shared-tilt}. For sequences sharing the same initial treatment, separately bounding their means and then subtracting can be conservative because both outcome arms must share one latent assignment mechanism. Theorem~\ref{thm:contraction} reduces optimization over this common assignment to the two feasible endpoints in each mixed cell.

For a binary-outcome example, omit $\X$ and take $p_0=1/4$, $p_1=1/2$, $\theta=1$. Then $w=1/4$ in both cells with $z_1=1$. In each cell let $q_1=1/2$, $\mu_{(1,1)}=3/4$ and $\mu_{(1,0)}=1/4$. As $\epsilon\downarrow0$, the sharp contrast interval has closure $[3/16,11/16]$. Separate differencing gives $[1/8,3/4]$, which is $25\%$ wider. The loss comes from assigning different optimizing propensities to the same latent stratum.

The response margins remain informative without a specified coupling. Instead let $p_0=1/10$, $p_1=1/2$ and $q_1=\mu_{(1,1)}=1/2$ in both cells with $z_1=1$. As $\epsilon\downarrow0$, the sharp region for $\psi(1,1)$ has closure $[7/16,9/16]$ at $\theta=1$. Leaving the coupling unspecified expands it to $[5/12,7/12]$, still narrower than $[1/4,3/4]$ for an unrestricted binary confounder. The response margins force the two alternative-response prevalences to sum to $1/5$, whereas the unrestricted confounder can have prevalence $1/2$ in both cells. Section~\ref{supp:sens-est} gives the calculations for both examples. More generally, Section~\ref{supp:unspecified-computation} optimizes over measurable couplings and gives an example where restricting the coupling to its endpoints misses the population bound. A finite grid of common odds ratios need not represent a covariate-dependent optimal coupling. This is a separate optimization from assignment-endpoint sufficiency at each fixed coupling.

Outcome mixtures underpin classical bounds in principal stratification \citep{zhangrubin2003}. For binary outcomes, the component limits yield a finite assignment search for principal contrasts.
\begin{corollary}[Binary principal contrasts]\label{cor:binary-principal-bounds}
Under Theorem~\ref{thm:unified-contrast}, suppose outcomes are binary and the
target stratum has positive mass. A component with mixing proportion $t>0$ in an
observed arm of mean $\mu$ has sharp mean limits
\[
\underline m(\mu,t)=\max\{0,1-(1-\mu)/t\},\qquad
\overline m(\mu,t)=\min\{1,\mu/t\}.
\]
For a principal contrast within first-period arm $j$, the conditional limits for type $v=1$ are
\[
\min_h\{\underline m(\mu_1,t_1(h))-\overline m(\mu_0,t_0(h))\},\qquad
\max_h\{\overline m(\mu_1,t_1(h))-\underline m(\mu_0,t_0(h))\}.
\]
For type $v=0$, these two optimizations use $1-t_a(h)$ in place of $t_a(h)$. For either type, it suffices to evaluate the feasible assignment endpoints and feasible solutions of $t_a(h)\in\{\mu_a,1-\mu_a\}$, $a=0,1$. Pure cells retain their observed means, and zero-score strata are omitted. Weight the resulting limits by $e_u(\X)/\E(e_u)$ to obtain principal-effect bounds.
\end{corollary}
The proof is in Section~\ref{supp:proof-binary-principal}.

\subsection{Sharpening bounds through effect restrictions}\label{sec:effect-restrictions}
The observed mixtures can remain compatible with large differences between response types' effects, giving a principal interval as wide as $[-1,1]$ for binary outcomes (Example~\ref{ex:principal-endpoint-failure}).
We restrict differences in conditional mean effects between response types.
The restriction applies within each response cell and covariate value.
It can sharpen principal bounds while allowing unequal outcome levels and treatment probabilities.
We specify a measurable function $\gamma_{j,s}(\X)\ge0$ and impose
\begin{equation}\label{eq:effect-heterogeneity-restriction}
|\delta_1-\delta_0|\le\gamma_{j,s}(\X).
\end{equation}
This restriction couples the two outcome arms, so the separate armwise
optimization in Theorem~\ref{thm:unified-contrast} no longer applies.
Theorem~\ref{thm:contraction} shows that the population
calculation nevertheless retains its two assignment evaluations.

At $\gamma=0$, $\delta_1=\delta_0=d$ makes the conditional contrast $\mu_1-\mu_0+(t_0-t_1)d$ independent of the latent label. The assignment-endpoint rule therefore also gives sharp bounds for each conditional principal effect. Equal conditional effects still permit unequal outcome levels and treatment probabilities across response types, so their common value need not be identified. Marginal principal and population bounds can differ because they average across response cells and covariate values differently. SPI (Assumption~\ref{ass:spi}) imposes the stronger restriction $\delta_0=\delta_1=0$.

For the symmetric observed configuration in Example~\ref{ex:principal-endpoint-failure}, the sharp conditional principal interval at $\gamma=0$ is
\[
\left[-\frac{1-2\epsilon}{2(1-\epsilon)},\,
\frac{1-2\epsilon}{2(1-\epsilon)}\right],
\]
whereas unrestricted principal effects range over $[-1,1]$. Supplementary Section~\ref{supp:effect-restrictions} gives the calculation.

For $\gamma>0$, a sharp conditional principal bound may instead be attained at an assignment probability $h$ strictly inside its feasible interval, even when the optimizing assignment is unique.

\begin{example}[A unique interior principal optimum]\label{ex:positive-budget-interior}
In a fixed response cell and covariate group, suppose $F_a=\operatorname{Bernoulli}(\mu_a)$ for $a\in\{0,1\}$, with $w=0.43$, $q=0.57$, $(\mu_0,\mu_1)=(0.17,0.83)$, $\epsilon=0.08$ and $\gamma=0.4$. Figure~\ref{fig:positive-budget-objectives} compares conditional upper bounds for the population effect $M_1-M_0$ and the principal effect $m_1(1)-m_0(1)$ for response type $V=1$. At each feasible assignment probability $h$, each curve maximizes its effect over compatible component outcome means subject to the effect budget. The population maximum occurs at an assignment endpoint. The principal maximum $0.905263$ occurs only at the interior assignment $h=0.750581$ and exceeds the best assignment-boundary value $0.895849$. Section~\ref{supp:effect-restrictions} gives the calculation and verifies that the unique interior optimum persists under small changes in the observed probabilities and means.
\end{example}

\begin{figure}[t]
\centering
\includegraphics[width=\linewidth]{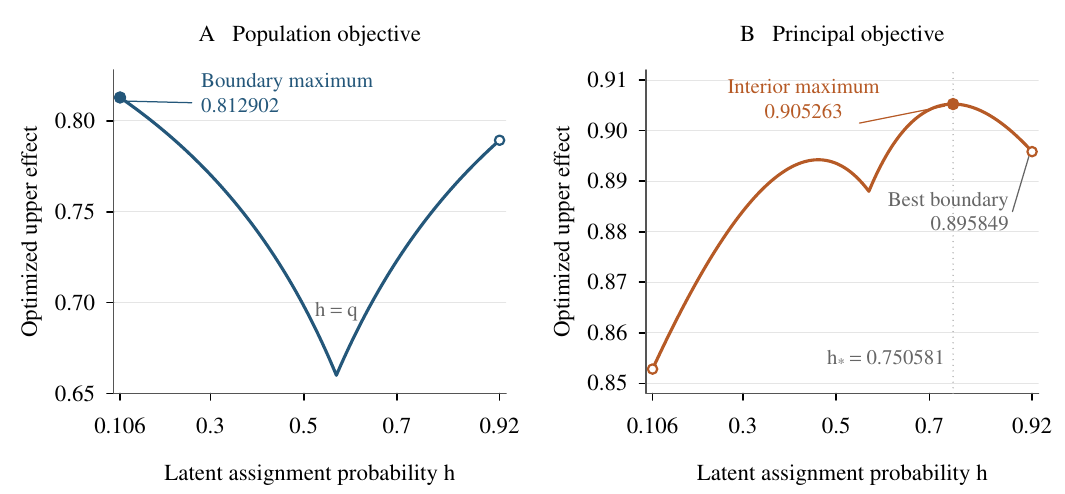}
\caption{Optimized conditional upper effects in Example~\ref{ex:positive-budget-interior}. At each $h$, both objectives maximize over compatible binary outcome means under the same budget $\gamma=0.4$. In panel A, the population maximum occurs at the assignment boundary $h=0.106047$. In panel B, the principal maximum for response type $V=1$ occurs at $h_\star=0.750581$ and exceeds both assignment-boundary values. Vertical scales differ.}\label{fig:positive-budget-objectives}
\end{figure}

For the effect budget in \eqref{eq:effect-heterogeneity-restriction}, the exact
feasible set at fixed $h$ is the polygon
\begin{equation}\label{eq:effect-heterogeneity-polygon}
\mathcal P_\gamma(h)=
\left(\mathcal D_{F_0}(t_0)\times\mathcal D_{F_1}(t_1)\right)
\cap\{(d_0,d_1):|d_1-d_0|\le\gamma\}.
\end{equation}
The linear expression in Theorem~\ref{thm:unified-contrast} is optimized over its
vertices, using only the two assignment endpoints for population targets.
The following result gives the corresponding exact search for principal targets.

\begin{proposition}[Finite optimization under an effect budget]\label{prop:finite-effect-budget}
Under Theorem~\ref{thm:unified-contrast}, suppose outcomes are binary and impose
\eqref{eq:effect-heterogeneity-restriction} with a finite nonnegative budget.
At fixed coupling and covariate value, each mixed cell's lower and upper linear target values
are attained by a finite candidate calculation. At each $h$, the polygon
\(\mathcal P_\gamma(h)\) has at most twelve listed vertices, consisting of
feasible rectangle corners and intersections of rectangle edges with
$d_1-d_0=\pm\gamma$.

Partition \(\mathcal H_\epsilon(w,q)\) at feasible solutions of
$t_a(h)\in\{\mu_a,1-\mu_a\}$.
For each resulting vertex formula, it suffices to evaluate feasible
partition endpoints, feasibility-equality roots and stationary roots of
its target value. Every nontrivial root equation has degree at most two
after clearing positive denominators. Isolated feasible points are retained,
and constant feasible pieces are represented by their boundaries.
Identically zero equations contribute no roots.

Integration gives the sharp scalar bounds. Principal bounds divide the
integrated numerator by the fixed positive stratum mass. Pure cells retain
their observed means. Population contrasts need only the two assignment
endpoints by Theorem~\ref{thm:contraction}.
\end{proposition}
The proof is in Section~\ref{supp:proof-finite-effect-budget}.
At a fixed coupling and covariate value, Proposition~\ref{prop:finite-effect-budget} turns each cell's binary optimization into a finite search over polygon vertices and candidate assignments. The vertices give the extreme compatible pairs of outcome-mean gaps at each assignment. The assignment candidates include interior stationary points and values where feasibility or a vertex formula changes, so the search retains principal optima that an assignment-endpoint search can miss.

\subsection{Supporting identification results}\label{sec:spi}
This section first gives point-identification conditions under LSI, then examines sensitivity to residual treatment selection within response profiles when LSI is relaxed.
For integrable potential outcomes, Assumptions~\ref{ass:sutva}--\ref{ass:positivity} and SPI (Assumption \ref{ass:spi}) identify the population mean below without specifying a response coupling. The principal identity additionally requires Assumption~\ref{ass:theta} and $\E\{e_u(\X)\}>0$.
\begin{align}
 \psi(j,a)&=\psi_g(j,a):=\E\left\{\sum_{s=0}^1p_j^s(\X)\mu_{(j,a)}(s,\X)\right\},\label{eq:gformula}\\
 \psi_u(j,a)&=\frac{\E\{e_u(\X)\mu_{(j,a)}(u_j,\X)\}}{\E\{e_u(\X)\}}.\label{eq:pce-id}
\end{align}
Zero-probability response cells are omitted. Corollaries~\ref{thm:union} and~\ref{thm:pce} in Supplementary Section \ref{supp:point} give these identities, with proofs in Sections~\ref{supp:proof-population-identification} and~\ref{supp:proof-principal-identification}. Homogeneous latent assignment means equal later-treatment probabilities across response types within each mixed cell $(\X,j,s)$, so $h_0=h_1=q$. Under LSI, this also identifies population means while allowing principal means to differ. Section~\ref{supp:point-identification-cases} gives the more general condition of zero products of latent assignment and outcome-mean differences.

\paragraph{Departures from response sufficiency.}\label{sec:residual-selection}\label{sec:sensitivity}
The effect budget $\gamma$ limits effect differences between profiles while maintaining LSI. It does not address an unmeasured trait that affects both later treatment and outcomes within a profile. Supplementary Section~\ref{app:residual} allows bounded residual selection in conditional means, retaining consistency, initial-treatment ignorability and positivity. Relaxing LSI changes which component laws are compatible with the observed outcome distributions, so the earlier identified intervals cannot simply be widened by a bias allowance.

\section{Inference for identified endpoints}\label{sec:estimation}\label{sec:bounds-inference-main}\label{sec:mr}\label{sec:dml}\label{sec:eifs}
We develop a common one-step construction for endpoints satisfying the smoothness and stability conditions below. For a smooth functional $T=\E\{g_\Gamma(\X)\}$, write its influence function as $\varphi_T=g_\Gamma-T+a_\Gamma$, where $\E(a_\Gamma\mid\X)=0$. Nuisance fits from outside each evaluation fold give
\begin{equation}\label{eq:onestep28}
\widehat T=\Pn\{g_{\widehat\Gamma}(\X)+a_{\widehat\Gamma}(\calO)\}.
\end{equation}
A binary endpoint makes the correction explicit. For the population
contrast $\psi(j,1)-\psi(j,0)$ at fixed $(\theta,\epsilon,\gamma)$,
write $q_s=\Prob(Z_2=1\mid\X,Z_1=j,S=s)$ and
$\mu_{sa}=\E(Y\mid\X,Z_1=j,S=s,Z_2=a)$.
The optimized lower or upper cell contrast bound from
Section~\ref{sec:effect-restrictions} is denoted by
$\mathcal B_s(p_0,p_1,q_s,\mu_{s0},\mu_{s1})$.
Its inputs are the response probabilities $p_0,p_1$, observed
later-treatment probability $q_s$ and observed outcome means
$\mu_{s0},\mu_{s1}$, all evaluated at $\X$.
Under these conditions, its complete derivatives give the following map, with
$\Gamma=(\pi,p,q,\mu)$ and conditioning on $\X$ suppressed.
\begin{equation}\label{eq:binary-endpoint-correction}
\begin{aligned}
g_\Gamma&=\sum_{s=0}^1p_j^s\mathcal B_s, 
\quad 
D_z=\one\{z=j\}(\mathcal B_1-\mathcal B_0)
       +\sum_{s=0}^1p_j^s\partial_{p_z}\mathcal B_s,\quad z\in\{0,1\},\\
a_\Gamma&=\sum_{z=0}^1D_z\frac{\one\{Z_1=z\}}{\pi_z}(S-p_z)+\sum_{s=0}^1\frac{\one\{Z_1=j,S=s\}}{\pi_j}
       \partial_{q_s}\mathcal B_s\,(Z_2-q_s)\\
&\quad+\sum_{s,a=0}^1
       \frac{\one\{Z_1=j,S=s,Z_2=a\}}
            {\pi_jq_s^a(1-q_s)^{1-a}}
       \partial_{\mu_{sa}}\mathcal B_s\,(Y-\mu_{sa}).
\end{aligned}
\end{equation}
The first term in $D_j$ differentiates the outer response weights.
Both response margins also enter through the latent prevalence
$w_s=e_{u^{(1)}}/p_j^s$. The derivatives retain this dependence and the
moving assignment or feasibility boundaries in the selected candidate.
The outer response weight $p_j^s$ cancels its factor in the residual denominators in the last two lines.

For a fixed principal stratum $u$, let $\mathcal B_u(\X)$ be its optimized conditional contrast bound under initial treatment $j$. It uses the same inputs in cell $s=u_j$. Its marginal endpoint is $T=N/D$, where $N=\E\{e_u(\X)\mathcal B_u(\X)\}$ and $D=\E\{e_u(\X)\}>0$. Applying the one-step construction separately to the numerator and mass gives cross-fitted estimates $\widehat N$ and $\widehat D$, with centered influence functions $\varphi_N$ and $\varphi_D$. The principal estimator and influence function are
\begin{equation}\label{eq:main-principal-ratio}
\widehat T=\frac{\widehat N}{\widehat D},\qquad
\varphi_T=\frac{\varphi_N-T\varphi_D}{D}.
\end{equation}
Differentiation includes both the principal score $e_u$ and the bound $\mathcal B_u$. Standard errors use the empirical variance of the estimated combined influence function, retaining numerator--mass covariance. Section~\ref{supp:distribution-scores} gives the common residual calculus and the equivalent ratio formula \eqref{eq:common-ratio}.

\paragraph{Binary endpoints.}
Corollary~\ref{cor:restricted-inference} gives joint inference for a fixed finite family of population and principal endpoints at sensitivity inputs chosen independently of the data. Its regularity conditions require a common interior domain, positive principal masses, persistent feasible candidates and bounded Lipschitz derivatives, including those of moving roots. Optimizers must remain separated or join with a common Lipschitz gradient. Uniform nuisance consistency preserves this domain. Section~\ref{supp:restricted-inference} states the precise conditions and the equality reduction required at $\gamma=0$.

For independent observations and a fixed number of cross-fitting folds with
nonvanishing proportions, let $\nu$ collect response margins, observed
later-treatment probabilities and outcome means. Foldwise, suppose
\[
r_n=\|\widehat\nu-\nu\|_2,\qquad
s_n=\|\widehat\pi-\pi\|_2,\qquad
r_n+s_n=o_{\Prob}(1),\qquad
r_n^2+s_nr_n=o_{\Prob}(n^{-1/2}).
\]
The one-step endpoint estimators then satisfy jointly
\begin{equation}\label{eq:main-binary-expansion}
\sqrt n(\widehat T_k-T_k)
=n^{-1/2}\sum_{i=1}^n\phi_k(\calO_i)+o_{\Prob}(1).
\end{equation}
Here $\phi_k$ is the complete endpoint influence function, including the principal numerator and mass correction. Their empirical influence covariance is consistent. Positive marginal variances
give asymptotically nominal pointwise normal intervals. For $M$ fixed endpoints,
the critical value $z_{1-\alpha/(2M)}$ gives simultaneous coverage at least
$1-\alpha$ asymptotically. Theorem~\ref{thm:distribution-inference} in Section~\ref{supp:distribution-inference} treats population endpoints with general outcome laws that admit the stated finite tail-integral representation and satisfy its regularity conditions.

These guarantees do not establish coverage for continuum sensitivity bands or a positivity margin chosen after inspecting fitted compatibility. Positive-probability switches with disagreeing derivatives are also excluded \citep{fangsantos2019}. General continuous principal bounds require a separate optimization regularity argument. Plug-in endpoints need a separate equivalence argument before using one-step standard errors.

Under SPI, the construction reduces to augmented principal-score estimation. Its influence function and rate conditions appear in Section~\ref{app:eif-pce}, alongside the standard longitudinal influence function for population means.

\section{Simulations}\label{sec:simulation-study}\label{sec:simulations}
\subsection{Design and identification geometry}
We assess the finite-sample performance of our estimators and confidence intervals for population and principal effect bounds. The targets are the lower and upper bounds of the population contrast $\psi(1,1)-\psi(1,0)$ and the principal contrast in $U=(1,1)$. Sample sizes are $n\in\{2000,8000,32000\}$.

The two main scenarios share a generating law satisfying LSI. Initial treatment is randomized with probability $1/2$, and the potential responses are independent Bernoulli variables with means $p_0=0.60$ and $p_1=0.65$, giving $\theta=1$. Within each recorded response cell, later treatment depends on the unobserved response $V=S(1-j)$, with $h_1>h_0$ chosen so that $(1-w)h_0+wh_1=q_{js}=0.70$. Binary potential outcomes are independent of assignment given $(\X,U)$, with means $0.04$ higher for $V=1$ than for $V=0$ in both treatment arms. These component means give observed means $\mu_{js0}=0.20$ and $\mu_{js1}=0.70$. Thus, the law has equal conditional effects while violating SPI. The latent positivity margin is $\epsilon=0.08$. Section~\ref{supp:sims} gives the complete generating process.

We present two scenarios that differ in the effect budget used to calculate the bounds. Their identified intervals follow from the finite optimization in Proposition~\ref{prop:finite-effect-budget}, evaluated at the true observed probabilities and averaged as described in Section~\ref{supp:sims}. Baseline sets $\gamma=0$, giving identical population and principal identified intervals $[-0.057971,0.649660]$. The equal-effect reduction in Section~\ref{sec:effect-restrictions} equates their conditional bounds, which here are constant across response cells and covariate values. Their different averaging weights therefore give the same marginal bounds, although their estimators can have different sampling behavior.

Interior sets $\gamma=1$, yielding population interval $[-0.086439,0.676501]$ and principal interval $[-0.310068,0.914516]$. In every target cell, $w=0.60$ and $\mathcal H_\epsilon=[83/150,23/25]$. The principal upper optimum occurs at
\[
h_\star=\frac12+\frac1{\sqrt{15}}\approx0.758199.
\]
Restricting this search to assignment endpoints gives $0.861088$, understating the sharp upper bound by $0.053428$. Population bounds retain the endpoint reduction. This illustrates why Proposition~\ref{prop:finite-effect-budget} includes interior assignments in the search for principal bounds.

All scenarios use synthetic covariates $(X_1,X_2)$ uniformly distributed on $\{-1,1\}^2$. These covariates do not affect the generating probabilities in Baseline and Interior. Section~\ref{supp:sims} gives the specifications and results for two additional scenarios. The Covariates scenario allows response probabilities, later-treatment probabilities and outcome means to vary with the covariates. The History scenario further introduces interactions involving treatment, response and covariates. Both additional scenarios use $\gamma=0$.

Across all scenarios, we use the cross-fitted one-step estimator in Section~\ref{sec:estimation} and construct nominal $95\%$ normal confidence intervals using estimated influence functions. Nuisance probabilities are estimated with the same ridge-logistic models used in the gym application. Comparisons with an oracle version using true nuisance functions assess the effect of nuisance estimation. Plug-in estimates based on the same fitted models assess the contribution of the one-step correction to point estimation. Section~\ref{supp:sims} gives fitting details and the construction of confidence regions for the full identified interval.

\subsection{Estimation and results}
Table~\ref{tab:simulation-reduced-main} reports evaluation metrics for both scenarios at all three sample sizes. Truth gives the identified endpoint. Each scenario and sample size uses $1000$ independently generated datasets.

\begingroup
\small
\setlength{\tabcolsep}{3pt}
\begin{table}[t]
\centering
\caption{Baseline and Interior scenarios at $n\in\{2000,8000,32000\}$. Every row uses the application nuisance learner and the same one-step normal interval, with $R=1000$ generated datasets per scenario and sample size. Truth is the identified endpoint. Width is mean interval length.}
\label{tab:simulation-reduced-main}
\begin{tabular}{@{}rlrrrrrr@{}}
\toprule
$n$ & Endpoint & Truth & Bias & RMSE & SD & Cov. (\%) & Width \\
\midrule
\multicolumn{8}{@{}l}{\textbf{Baseline}} \\
2,000 & Pop. L & -0.0580 & -0.020 & 0.041 & 0.036 & 92.4 & 0.148 \\
2,000 & Pop. U & 0.6497 & 0.008 & 0.024 & 0.023 & 88.8 & 0.078 \\
2,000 & Princ. L & -0.0580 & -0.024 & 0.056 & 0.050 & 93.7 & 0.205 \\
2,000 & Princ. U & 0.6497 & 0.006 & 0.028 & 0.027 & 92.1 & 0.097 \\
\addlinespace[2pt]
8,000 & Pop. L & -0.0580 & -0.005 & 0.014 & 0.013 & 93.6 & 0.052 \\
8,000 & Pop. U & 0.6497 & 0.000 & 0.010 & 0.010 & 94.9 & 0.038 \\
8,000 & Princ. L & -0.0580 & -0.003 & 0.017 & 0.017 & 95.5 & 0.066 \\
8,000 & Princ. U & 0.6497 & -0.000 & 0.012 & 0.012 & 96.0 & 0.047 \\
\addlinespace[2pt]
32,000 & Pop. L & -0.0580 & -0.001 & 0.007 & 0.006 & 93.6 & 0.025 \\
32,000 & Pop. U & 0.6497 & -0.000 & 0.005 & 0.005 & 95.8 & 0.019 \\
32,000 & Princ. L & -0.0580 & -0.001 & 0.008 & 0.008 & 95.6 & 0.031 \\
32,000 & Princ. U & 0.6497 & -0.000 & 0.006 & 0.006 & 95.7 & 0.023 \\
\midrule
\multicolumn{8}{@{}l}{\textbf{Interior}} \\
2,000 & Pop. L & -0.0864 & -0.023 & 0.048 & 0.043 & 93.2 & 0.170 \\
2,000 & Pop. U & 0.6765 & -0.000 & 0.019 & 0.019 & 94.6 & 0.072 \\
2,000 & Princ. L & -0.3101 & -0.041 & 0.090 & 0.080 & 92.9 & 0.268 \\
2,000 & Princ. U & 0.9145 & -0.003 & 0.037 & 0.037 & 89.2 & 0.126 \\
\addlinespace[2pt]
8,000 & Pop. L & -0.0864 & -0.003 & 0.016 & 0.015 & 95.4 & 0.059 \\
8,000 & Pop. U & 0.6765 & 0.000 & 0.009 & 0.009 & 95.2 & 0.036 \\
8,000 & Princ. L & -0.3101 & -0.008 & 0.023 & 0.022 & 94.1 & 0.071 \\
8,000 & Princ. U & 0.9145 & 0.000 & 0.016 & 0.016 & 95.0 & 0.063 \\
\addlinespace[2pt]
32,000 & Pop. L & -0.0864 & -0.000 & 0.006 & 0.006 & 94.7 & 0.025 \\
32,000 & Pop. U & 0.6765 & -0.000 & 0.005 & 0.005 & 95.2 & 0.018 \\
32,000 & Princ. L & -0.3101 & -0.000 & 0.006 & 0.006 & 95.2 & 0.023 \\
32,000 & Princ. U & 0.9145 & -0.001 & 0.008 & 0.008 & 95.3 & 0.031 \\
\bottomrule
\end{tabular}
\end{table}
\endgroup

Overall, both scenarios exhibit small bias and RMSE at larger sample sizes, with coverage close to the nominal $95\%$ level for most endpoints. RMSE and interval width decrease as sample size increases. Some intervals undercover at smaller sizes. Baseline population upper coverage improves from $88.8\%$ at $n=2000$ to $95.8\%$ at $n=32000$, and Interior principal upper coverage from $89.2\%$ to $95.3\%$. 

The one-step correction accounts for nuisance estimation in endpoint inference under the conditions of Section~\ref{sec:estimation}. Its finite-sample bias and RMSE need not improve on those of plug-in estimation, as Table~\ref{tab:simulation-method-comparison} in Section~\ref{supp:sims} illustrates. The full-grid paired comparisons documented in Section~\ref{supp:sims} show larger mean squared error than plug-in estimation for all $16$ endpoints at $n=2000$. At $n=32000$, it is smaller for all eight population endpoints but remains larger for seven of eight principal endpoints, with Interior upper the exception. Comparisons with saturated nuisance estimates show that the observed deficits are not confined to the ridge-logistic learner. True-influence remainders decrease with sample size but remain material for the History population lower endpoint, where negative fixed-fit drift persists (Section~\ref{supp:diagnostic-replay}). These comparisons identify finite-sample limitations of the one-step approximation.

Identification uncertainty accounts for most of the Interior principal outward-region width. Its true identified width is $1.2246$. At $n=32000$, the mean region width $1.2552$ decomposes into a mean estimated endpoint span of $1.2242$ and mean Bonferroni padding of $0.0310$. Whole-set enclosure is $98.3\%$, a different criterion from calibrated two-sided inference for each endpoint or containment of the generating causal effect.

Table~\ref{tab:simulation-reduced-detail} in Section~\ref{supp:sims} gives additional results for the Covariates and History scenarios. Bias magnitude and RMSE again decrease with sample size. At larger sample sizes, coverage is close to the nominal $95\%$ level for most endpoints, with modest undercoverage persisting for some lower endpoints. At $n=32000$, population lower coverage is $93.2\%$ in Covariates and $92.2\%$ in History, whose principal lower coverage is $92.9\%$.

\section{Gym attendance after an incentive program}\label{sec:application}

\citet{milkman2021} conducted a field experiment at a US fitness chain to test digital programs for increasing gym attendance and assess whether any increases persisted after the programs ended. We focus on a four-week program combining exercise planning, text reminders and small rewards for gym visits. Our analysis includes $1{,}932$ members retained in the released data, with intended assignment to this program ($Z_1=1$) or its placebo comparison ($Z_1=0$). Section~\ref{supp:gym-design} defines the cohort and attendance coding. The response $S$, later action $Z_2$ and outcome $Y$ indicate recorded gym attendance in weeks 1, 2 and 5. The program operates in weeks 1--4.

Our population target is the difference in week-5 attendance probabilities between at least one gym visit and no visits in week 2 among retained members, with initial assignment fixed at $Z_1=1$. The principal target is the same contrast among always-attenders, $U=(1,1)$, who would attend in week 1 under either initial assignment. The program continues under both interventions. Consistency in Assumption~\ref{ass:sutva} is imposed on binary attendance, so we assume that differences in visit number and timing among attendees do not alter a member's week-5 outcome.  Randomization supports Assumption~\ref{ass:lsi}(a), provided exchangeability is preserved after retention. Assumption~\ref{ass:lsi}(b) allows confounding through latent attendance tendencies summarized by $U$ and baseline covariates, including prior attendance, but excludes residual confounding by motivation or schedules within these groups. For example, unrecorded commitment to regular exercise may affect both week-2 and week-5 attendance. The profile $U$ may summarize this commitment by distinguishing members who would attend initially under either assignment from those whose attendance would change with assignment. Randomization provides initial positivity before retention. Both week-2 attendance states occur in every observed $(Z_1,S)$ group, but Assumption~\ref{ass:positivity} requires conditional positivity in the retained population, including within latent response profiles. Section~\ref{supp:emp-policy} details the interventions and identifying assumptions. We use a latent positivity margin of $\epsilon=0.005$, selected after checking compatibility with the fitted law (Section~\ref{supp:emp-nuisance}).

Cross-fitted ridge-logistic models adjust for demographics, recruitment cohort, enrollment date and baseline attendance, with treatment-history interactions (Section~\ref{supp:emp-nuisance}). At $\theta=1$, the plug-in SPI effects are $18.62$ percentage points for the population and $23.24$ for always-attenders. These are the reference values in Figure~\ref{fig:gym-gamma}, computed from the same fitted law as the identification regions. They do not depend on the positivity margin. Under SPI, the nominal pointwise $95\%$ confidence intervals based on augmented estimates lie entirely above zero for both targets (Table~\ref{tab:gym-spi}).

\begin{figure}[t]
\centering
\includegraphics[width=\textwidth]{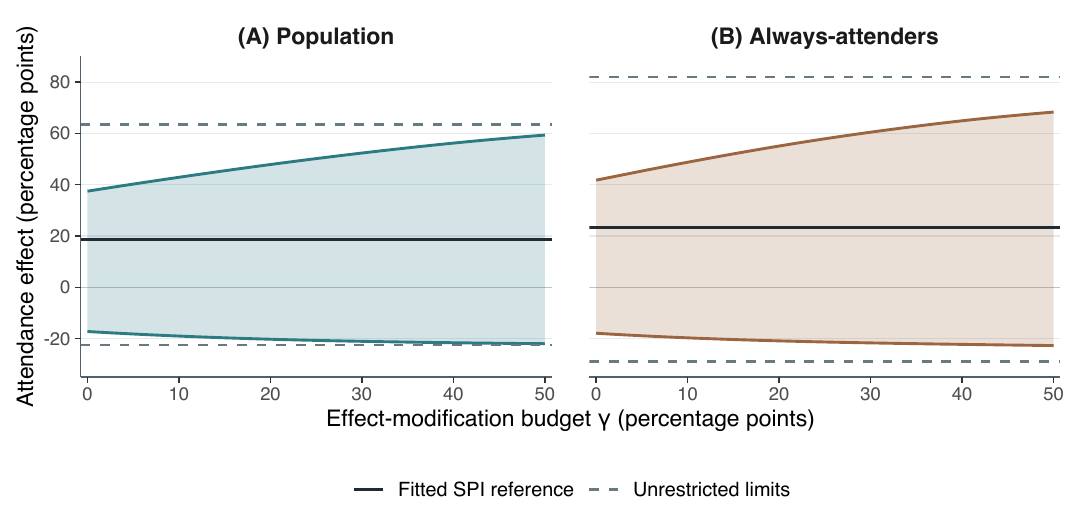}
\caption{Effect restrictions and attendance persistence. Shading gives fitted-law identification regions for the retained population (panel A) and always-attenders (panel B) at $\theta=1$ and $\epsilon=0.005$. Both regions widen with $\gamma$ and include zero even at $\gamma=0$. The budget $\gamma$ is displayed in percentage points. Solid horizontal lines are plug-in SPI values for the same fitted law, and dashed lines are unrestricted limits. The shading does not represent sampling uncertainty.}\label{fig:gym-gamma}
\end{figure}

Figure~\ref{fig:gym-gamma} shows how the population and principal identification regions widen as the effect budget $\gamma$ increases. At $\gamma=0.10$, the effect of week-2 attendance on week-5 attendance can differ by at most ten percentage points between the two response types in each cell and covariate group. The types may still have different outcome probabilities under either action. The choice of $\gamma$ requires subject-matter judgment. LSI remains imposed, so the budget does not address residual confounding within profiles.

At $\theta=1$ and $\gamma=0.10$, the population and principal regions are $[-19.01,42.89]$ and $[-19.69,48.76]$ percentage points, respectively, with widths $28\%$ and $38\%$ below their unrestricted values. These restrictions substantially narrow the compatible effects but leave both increases and decreases in week-5 attendance possible for either target. Even at $\gamma=0$, the regions remain $[-17.23,37.48]$ and $[-17.90,41.76]$. Equal conditional effects still permit selection between types with different outcome probabilities, so they do not determine the effect's sign. SPI additionally equates the types' outcome probabilities under each action, yielding the positive point estimates above.

Panel (B) of Figure~\ref{fig:gym-gamma} uses the full principal search, including interior assignment candidates. Its dashed lines show the unrestricted region $[-28.97,81.88]$ percentage points at $\theta=1$. For comparison, an endpoint-only search would give $[-25.55,75.35]$. This shortcut would exclude effects compatible with the same observed data and assumptions, understating the identified interval's width by $9.94$ percentage points. Interior assignments are therefore needed for sharp principal bounds in this application. For the population contrast, enforcing a common assignment reduces interval width by $0.71$ percentage points relative to separately bounding and subtracting the two means.

Joint changes in response coupling and effect budget retain both signs throughout the grid in Table~\ref{tab:gym-joint-sensitivity}. Different couplings imply different proportions and covariate distributions of always-attenders, so principal comparisons also depend on the assumed composition of that stratum. The displayed grid of common odds ratios does not optimize over all covariate-dependent couplings (Section~\ref{supp:emp-joint-sensitivity}).
Section~\ref{supp:emp-inference} gives the additional one-step endpoint analysis at $\gamma=0$, $0.05$ and $0.10$. All six nominal simultaneous outward regions include zero. These regions incorporate sampling uncertainty in the estimated endpoints. 

\section{Discussion}\label{sec:discussion}
An open question is whether the assignment-endpoint reduction extends to restrictions linking different response cells. For example, equal latent outcome-mean gaps across cells may not be preserved when each cell uses a different contraction factor. A characterization of restrictions that permit coordinated contractions would clarify the scope of this extension. Longer treatment histories pose a related challenge because several assignment mechanisms must remain compatible with one joint response profile. A further direction is inference for optimized bounds when competing candidates meet with different derivatives. Confidence procedures that remain valid near such switches would accommodate settings where small changes in the estimated law can change which latent model determines a bound.

\bibliographystyle{Chicago}\bibliography{literature}

\clearpage
\setcounter{page}{1}
\normalsize
\begin{center}
{\Large\bfseries Supplementary material}\\[4pt]
{\normalsize Sharp bounds for population and principal effects under latent sequential ignorability}\\[3pt]
Yuki Ohnishi
\end{center}
\allowdisplaybreaks[1]
\setcounter{section}{0}
\renewcommand{\thesection}{S\arabic{section}}
\renewcommand{\theHsection}{supp.S\arabic{section}}
\setcounter{equation}{0}
\renewcommand{\theequation}{S\arabic{equation}}
\renewcommand{\theHequation}{supp.S\arabic{equation}}
\setcounter{table}{0}
\renewcommand{\thetable}{S\arabic{table}}
\renewcommand{\theHtable}{supp.S\arabic{table}}
\setcounter{figure}{0}
\renewcommand{\thefigure}{S\arabic{figure}}
\renewcommand{\theHfigure}{supp.S\arabic{figure}}
\setcounter{theorem}{0}\renewcommand{\thetheorem}{S\arabic{theorem}}\renewcommand{\theHtheorem}{supp.S\arabic{theorem}}
\setcounter{lemma}{0}\renewcommand{\thelemma}{S\arabic{lemma}}\renewcommand{\theHlemma}{supp.S\arabic{lemma}}
\setcounter{proposition}{0}\renewcommand{\theproposition}{S\arabic{proposition}}\renewcommand{\theHproposition}{supp.S\arabic{proposition}}
\setcounter{corollary}{0}\renewcommand{\thecorollary}{S\arabic{corollary}}\renewcommand{\theHcorollary}{supp.S\arabic{corollary}}
\setcounter{assumption}{0}\renewcommand{\theassumption}{S\arabic{assumption}}\renewcommand{\theHassumption}{supp.S\arabic{assumption}}
\setcounter{remark}{0}\renewcommand{\theremark}{S\arabic{remark}}\renewcommand{\theHremark}{supp.S\arabic{remark}}

\setcounter{example}{0}\renewcommand{\theexample}{S\arabic{example}}\renewcommand{\theHexample}{supp.S\arabic{example}}

\section{Compatibility and sharp bounds}\label{supp:compatibility}\label{app:identification}\label{supp:bounds}
\subsection{Mixtures, bias and outcome compatibility}\label{sec:bias}\label{app:lemma-mixture}\label{app:bias}\label{sec:outcome-compatibility}\label{sec:full-law-bounds}\label{app:geometry}
For a positive-probability cell $(j,s)$, write $h_u^{(a)}=h_u^a(1-h_u)^{1-a}$. In a mixed cell, the assignment and outcome-mean gaps are
\begin{equation}\label{eq:deltarho}
 \rho(j,s,\X)=h_{u^{(1)}}(j,\X)-h_{u^{(0)}}(j,\X),\qquad
 \delta_{(j,a)}(j,s,\X)=m_{(j,a)}(u^{(1)},\X)-m_{(j,a)}(u^{(0)},\X).
\end{equation}
All sums below omit zero-score strata. Mean statements in this subsection assume $\E|Y|<\infty$.
\begin{lemma}[Selected mixture under LSI]\label{lem:mixture}
Under Assumptions~\ref{ass:sutva}--\ref{ass:positivity},
\begin{align}
 \Prob(U=u\mid\X,Z_1=j,S=s)&=w_u(j,s,\X),\label{eq:post-cell}\\
 q_1(j,s,\X)&=\sum_u w_uh_u,\label{eq:q-mixture}\\
 \Prob(U=u\mid\X,Z_1=j,S=s,Z_2=a)&=\widetilde w_u^{\,a}=w_uh_u^{(a)}/q_a,\label{eq:post-tilted}\\
 \mu_{(j,a)}(s,\X)&=\sum_u\widetilde w_u^{\,a}m_{(j,a)}(u,\X).\label{eq:tilted-mixture}
\end{align}
The last identity also holds for conditional outcome laws without a moment assumption.
\end{lemma}
The proof is in Section~\ref{supp:proof-mixture}.
\begin{proposition}[Bias correction and assignment feasibility]\label{prop:bias}\label{prop:sens}
Under the same assumptions, the following hold.
\begin{enumerate}[label=(\roman*),nosep]
\item The g-formula bias is
\begin{equation}\label{eq:bias}
 B_{\zbar}:=\psi_g(\zbar)-\psi(\zbar)
 =\E\sum_s p_{z_1}^s\frac{(-1)^{1-z_2}w(1-w)\delta_{\zbar}\rho}{q_{z_2}}.
\end{equation}
Its cell integrand is $\Cov_w(m_{\zbar},h^{(z_2)})/q_{z_2}$. Thus fixed $(\delta,\rho,\theta)$ identify $\psi=\psi_g-B_{\zbar}$ whenever the compatible model is nonempty.
\item In a mixed cell, $h_{u^{(1)}}=q_1+(1-w)\rho$ and $h_{u^{(0)}}=q_1-w\rho$. Both lie in $[\epsilon,1-\epsilon]$ exactly when
\[
 \max\left\{\frac{\epsilon-q_1}{1-w},\frac{q_1-1+\epsilon}{w}\right\}
 \le\rho\le
 \min\left\{\frac{1-\epsilon-q_1}{1-w},\frac{q_1-\epsilon}{w}\right\}.
\]
This restricts assignment only. Outcome compatibility is imposed below. In a pure cell the bias is zero and only the observed propensity is constrained.
\end{enumerate}
\end{proposition}
The proof is in Section~\ref{supp:proof-bias}.

\paragraph{Outcome compatibility.}
Theorem~\ref{thm:unified-contrast} gives the following equivalent sensitivity formulas. For the target arm's conditional law $F$ and selected component share $t$, write
\begin{equation}\label{eq:trimmed-means}
T_F^-(t)=L_F(t)/t,\qquad T_F^+(t)=U_F(t)/t.
\end{equation}
\begin{proposition}[Compatible sensitivity inputs]\label{prop:full-law-compatibility}
For a measurable response coupling reproducing the observed margins and measurable inputs $(\delta_{\zbar},\rho)$, suppose first-treatment probabilities lie in $[\epsilon,1-\epsilon]$, $\epsilon>0$, and $\E|Y|<\infty$. A compatible LSI law exists if and only if the assignment interval in Proposition~\ref{prop:sens}(ii) holds and, in every mixed cell and target arm,
\begin{equation}\label{eq:outcome-compatibility}
\frac{T_F^-(t)-\mu}{1-t}\le\delta_{\zbar}\le
\frac{T_F^+(t)-\mu}{1-t}.
\end{equation}
Pure cells require only their observed treatment probabilities in $[\epsilon,1-\epsilon]$. For binary outcomes, \eqref{eq:outcome-compatibility} is
\begin{equation}\label{eq:binary-delta-compatibility}
\max\left\{-\frac{\mu}{1-t},-\frac{1-\mu}{t}\right\}
\le\delta_{\zbar}\le
\min\left\{\frac{1-\mu}{1-t},\frac{\mu}{t}\right\}.
\end{equation}
\end{proposition}
The proof is in Section~\ref{supp:proof-full-law-compatibility}.

\paragraph{Sequence bounds.}
Under the assumptions of Theorem~\ref{thm:unified-contrast}, the contraction in Theorem~\ref{thm:contraction} gives the following sequence-mean specialization. For the target arm's feasible tilt interval $[t^-_\epsilon,t^+_\epsilon]$ and deviations $E_F^\pm$ in \eqref{eq:trimmed-deviations}, define
\[
D^\pm_{F,\epsilon}=\max\{E_F^\pm(t^-_\epsilon),E_F^\pm(t^+_\epsilon)\},
\]
with $D^\pm=0$ in pure cells. The sharp sequence-mean interval is
\begin{equation}\label{eq:full-law-region}
\left[\psi_g(\zbar)-\E\sum_s p_{z_1}^s D^-_{F_{\zbar,s,\X},\epsilon},\quad
\psi_g(\zbar)+\E\sum_s p_{z_1}^s D^+_{F_{\zbar,s,\X},\epsilon}\right].
\end{equation}
An unspecified coupling gives the union over measurable couplings reproducing the margins. The endpoint argument is included in Section~\ref{supp:proof-shared-tilt}.

\subsection{Support bounds as a specialization}\label{supp:sens-est}\label{app:bounds}
For a fixed target cell abbreviate $(w,q_{z_2},\mu_{\zbar})$ by $(w,q,\mu)$, and let $A=y_{\max}-\mu$, $B=\mu-y_{\min}$. Under only a mixture mean and support $[y_{\min},y_{\max}]$, the exact component restriction at $0<t<1$ is
\[
 \max\{y_{\min},\mu-(1-t)A/t\}\le m_1
 \le\min\{y_{\max},\mu+(1-t)B/t\},\qquad
 m_0=\frac{\mu-tm_1}{1-t}.
\]
It is an outer constraint for a general observed outcome law, and coincides with its trimmed-mean constraint when that law is supported on the two endpoints. Define
\begin{equation}\label{eq:envelope}
 \begin{split}
 f(t)&=\min\left\{\frac{(w-t)A}{1-t},\left(\frac wt-1\right)B\right\},\quad t\le w,\\
 g(t)&=\min\left\{\frac{(t-w)B}{1-t},\left(1-\frac wt\right)A\right\},\quad t\ge w.
 \end{split}
\end{equation}
The boundary conventions are $f(0)=wA$ and $g(1)=(1-w)A$.

Let $\epsilon^*>0$ be the essential infimum of
$\min\{\pi_1,1-\pi_1,q_1(j,s,\X),1-q_1(j,s,\X)\}$ over both initial arms and supported response cells. At $0\le\epsilon\le\epsilon^*$ the feasible target-arm tilt interval is
\[
 t^-_\epsilon=\max\left\{\frac{w\epsilon}{q},\frac{q-(1-\epsilon)(1-w)}q\right\},\qquad
 t^+_\epsilon=\min\left\{\frac{w(1-\epsilon)}q,\frac{q-\epsilon(1-w)}q\right\}.
\]
It contains $w$ and, for mixed cells and positive $\epsilon$, lies inside $(0,1)$. Write $t^\pm=t^\pm_0$ and
\[
 D^+_{\zbar,\epsilon}=\max\{f(t^-_\epsilon),g(t^+_\epsilon)\},\qquad
 I_\epsilon=\left[\psi_g-\E\sum_s p_{z_1}^sD^-_{\zbar,\epsilon},\,
                  \psi_g+\E\sum_s p_{z_1}^sD^+_{\zbar,\epsilon}\right],
\]
where $D^-_{\zbar,\epsilon}$ interchanges $A,B$. Put $D^\pm_{\zbar}=D^\pm_{\zbar,0}$ and set both deviations to zero in pure cells. The tilt interval is equivalent to the assignment interval of Proposition~\ref{prop:sens}(ii).

\begin{proposition}[Support specialization]\label{prop:bounds}
Under consistency, LSI, the specified coupling, bounded observed outcomes and $\epsilon^*>0$, the following hold.
\begin{enumerate}[label=(\roman*),nosep]
\item At each positive latent margin, $I_\epsilon$ contains the sharp full-law interval. Equality holds for an observed law on the two support endpoints. Also $I_\epsilon\subseteq I_{\epsilon'}\subseteq I_0$ for $0\le\epsilon'\le\epsilon\le\epsilon^*$.
\item If $y_{\min}<\mu<y_{\max}$ almost surely in every target response cell, then $I_\epsilon\to I_0$ as $\epsilon\downarrow0$. In mixed cells, its zero-margin extrema require a latent propensity zero or one. For endpoint-supported outcomes, $I_0$ is the closure of the identified set over positive margins, and this set is the open interval with those endpoints when mixed cells have positive probability.
\item The cellwise restriction $\delta_{\zbar}\rho=0$, including SPI or LPH, collapses the identified set to $\{\psi_g\}$.
\end{enumerate}
\end{proposition}
The proof is in Section~\ref{supp:proof-support-bounds}.
Zero-margin sharpness in part (ii) refers to the closure of the identified set over positive margins, under the stated condition that observed means lie strictly inside the outcome support. At a support endpoint, the formal zero-mass boundary values need not equal the limits from positive component shares or give that closure.

\begin{remark}\label{rem:bounds}
The additional identifying restrictions in Proposition~\ref{prop:bounds}(iii) impose no further observed-law test, by construction L5 in Section~\ref{app:relations}. Shared-arm contrasts use Proposition~\ref{prop:shared-tilt}. Separate mean differencing can be conservative. For binary contrasts, the zero-margin closure holds under interior means in both arms. Each population mean or contrast interval contains its corresponding g-formula value, so unions over measurable couplings remain intervals.
\end{remark}

\paragraph{The two numerical examples.}
For $w=1/4$, $q=1/2$, $\mu_1=3/4$ and $\mu_0=1/4$, the limiting deviations are
\[
\begin{array}{c|cccc}
h&E_1^+&E_1^-&E_0^+&E_0^-\\\hline
0&1/16&3/16&1/8&1/8\\
1&1/8&1/8&3/16&1/16
\end{array}
\]
Thus $\Delta_g=1/2$, $J^+=3/16$ and $J^-=5/16$, giving $[3/16,11/16]$. Separate optimization gives deviations $1/4$ and $3/8$, hence $[1/8,3/4]$. Interior binary means ensure these are limits of positive-margin bounds.

For the response-margin example, write $a=e_{(1,1)}\in[0,1/10]$. The two alternative response prevalences are $2a$ and $1/5-2a$. At $q=\mu=1/2$ and $0\le w\le1/5$, the limiting deviation in either direction is
\[
d(w)=\frac{w}{2(1-2w)}.
\]
Independence sets $a=1/20$, giving $d(1/10)=1/16$. Convexity of $d$ makes $\{d(2a)+d(1/5-2a)\}/2$ largest at $a=0$ or $1/10$, where it is $1/12$. The intervals share centre $1/2$, so the coupling union has closure $[5/12,7/12]$. Without response restrictions, consistency gives $[q\mu,q\mu+1-q]=[1/4,3/4]$. A binary confounder of prevalence $1/2$, assignment propensities tending to zero and one, and Bernoulli component outcomes attain those limits.

\paragraph{A strict support-and-mean relaxation.}\label{supp:full-law-strictness}
For $w=q=1/2$, $\epsilon=1/4$ and $F_a=\operatorname{Unif}[0,1]$, the feasible assignment interval is $[1/4,3/4]$ and $t_1=h$. Since $L_F(t)=t^2/2$ and $U_F(t)=t-t^2/2$, the compatible gap is $\mathcal D_F(t)=[-1/2,1/2]$. The sharp population cell mean interval is therefore
\[
\left\{\frac12+\left(\frac12-h\right)d:
 h\in[1/4,3/4],\ d\in[-1/2,1/2]\right\}=[3/8,5/8].
\]
At $h=1/4$, allocating the upper or lower quarter of the uniform law to $V=1$ attains the endpoints and preserves the observed mixture. Retaining only its support and mean instead allows
$|d|\le\min\{1/(2h),1/[2(1-h)]\}$, giving $[1/3,2/3]$.
For example, $(m_1,m_0)=(1,1/3)$ at $h=1/4$ gives $2/3$, but cannot mix to the uniform law, whose upper-quarter mean is at most $7/8$. Complementing outcomes gives the lower endpoint. With no covariates, $\pi_0=\pi_1=p_0=p_1=q=1/2$, $\theta=1$ and uniform outcomes in every cell, Theorem~\ref{thm:unified-contrast} embeds the sharp interval $[3/8,5/8]$ in a longitudinal law. The wider support-and-mean interval discards the uniform-law constraint.

\subsection{Matched comparisons with sensitivity models}\label{supp:matched-comparisons}

\begin{table}[t]
\centering
\caption{Relationship to the closest identification and sensitivity results, compared at the same observed cell law. No general ordering of the model classes is asserted.}\label{tab:contribution-comparison}
\small
\begin{tabular}{@{}>{\raggedright\arraybackslash}p{.23\textwidth}>{\raggedright\arraybackslash}p{.33\textwidth}>{\raggedright\arraybackslash}p{.38\textwidth}@{}}
\toprule
Ingredient & Related results & Role in this paper\\
\midrule
Latent sequential selection
& Principal profiles account for later assignment \citep{ricciardi2020}.
& Full-law compatibility without parametric component outcome models.\\[1ex]
Response coupling and principal weighting
& Odds-ratio coupling and identification under mean principal ignorability \citep{tong2025}.
& Response margins constrain latent prevalence while component outcome means may differ.\\[1ex]
Outcome and assignment sensitivity
& Sequential density-ratio or assignment and outcome restrictions \citep{tan2025,zhangxutan2026,zhangxutandist2025}.
& Population marginal laws admit assignment-endpoint contraction under local gap restrictions preserved by common scaling.\\[1ex]
Principal-effect sensitivity
& Bounds for initial-treatment principal effects under departures from principal ignorability \citep{chen2026}.
& Later-treatment principal effects can require interior assignment, whereas population laws admit endpoint contraction.\\
\bottomrule
\end{tabular}
\end{table}
Fixed-share mixture trimming is established \citep{lee2009,dornguo2023,dornguokallus2025}. Binary-confounder sensitivity models \citep{rosenbaumrubin1983} and mean-bias decompositions \citep{vanderweelearah2011} also have established foundations. The present comparison fixes the observed cell law $(q,F_0,F_1)$ and response-profile prevalence $w$. The shared assignment $h_v=\Prob(Z_2=1\mid V=v,\X,Z_1=j,S=s)$ must satisfy
\begin{equation}\label{eq:comparison-common-profile}
\begin{gathered}
 q=wh_1+(1-w)h_0,\qquad
 t_1=\frac{wh_1}{q},\qquad
 t_0=\frac{w(1-h_1)}{1-q},\\
 F_a=t_aF_{a1}+(1-t_a)F_{a0}.
\end{gathered}
\end{equation}
In the enhanced marginal sensitivity model of \citet{zhangxutan2026}, the corresponding assignment
density ratio, after identifying its unrestricted confounder with $V$, is
\[
 \lambda(v)=\frac{\Prob(V=v\mid Z_2=0)}{\Prob(V=v\mid Z_2=1)}
 =\frac{q(1-h_v)}{(1-q)h_v},
\]
with the cell conditioning suppressed. Our positivity margin corresponds to
\[
\frac{q\epsilon}{(1-q)(1-\epsilon)}\le\lambda(v)\le
\frac{q(1-\epsilon)}{(1-q)\epsilon}.
\]
Matching this range does not impose $qt_1+(1-q)t_0=w$ at the response-profile prevalence. The enhanced model constrains $m_a(v)-\mu_a=(v-t_a)\delta_a$ separately by action, whereas the effect budget bounds $|\delta_1-\delta_0|$. Thus $\gamma=0$ permits $\delta_0=\delta_1\ne0$, while zero widths for both enhanced-model mean constraints force both gaps to zero.

The outcome-density ratios in \citet{tan2025} need not admit the common binary decomposition in \eqref{eq:comparison-common-profile}. The distributionally enhanced model of \citet{zhangxutandist2025} bounds each latent outcome density relative to its observed treatment-arm density, without fixing response-profile prevalence or the present effect budget. Other relaxations also give distributional bounds \citep{mastenpoirier2025}.

\citet{chen2026} bound initial-treatment principal effects through stratum--confounder imbalance and outcome relative risks. Their observed-cell mixture weights are principal-score fractions, analogous to $w$ before later-treatment selection in our model. Their setup has no later treatment. Here later-treatment selection replaces $w$ by $t_a(h)$. Theorem~\ref{thm:contraction} addresses this shared assignment, preserving population marginal laws under common contraction while changing principal laws. The different restrictions do not establish a general ordering of the resulting regions.

\subsection{Sharp bounds under bounded effect heterogeneity}
\label{supp:effect-restrictions}
Write $\ell_a=\ell_{F_a}(t_a)$ and $u_a=u_{F_a}(t_a)$. All assumptions,
pure-cell conventions and principal normalizations are those of
Theorem~\ref{thm:unified-contrast}. The restrictions below are imposed
separately in each mixed cell and covariate value.

\begin{corollary}[Compatible effect restrictions]
\label{cor:effect-heterogeneity-bounds}
Under \eqref{eq:effect-heterogeneity-restriction}, replace the inner armwise
optimizations in \eqref{eq:unified-contrast-bounds} by
\[
\min_{\bm d\in\mathcal P_\gamma(h)}\sum_{a=0}^1 B_a(h)d_a,
\qquad
\max_{\bm d\in\mathcal P_\gamma(h)}\sum_{a=0}^1 B_a(h)d_a.
\]
The resulting scalar identified interval is sharp, with every value attained.
The feasible vertex list consists of the rectangle corners and the
intersections of its four edges with $d_1-d_0=\gamma$ or
$d_1-d_0=-\gamma$, retaining only feasible points. It has at most twelve
listed points, before removing duplicates. Population bounds require only
$h\in\{h^-_\epsilon,h^+_\epsilon\}$.

The intervals are nested as $\gamma$ increases. At $\gamma=0$ the optimization
uses $d_0=d_1=d\in[\max(\ell_0,\ell_1),\min(u_0,u_1)]$.
At a fixed $h$, the restriction leaves the entire feasible rectangle unchanged exactly when
$\gamma\ge\max(u_1-\ell_0,u_0-\ell_1)$.
In particular, $\gamma\ge2(y_+-y_-)$ recovers the unrestricted model.
\end{corollary}
The proof is in Section~\ref{supp:proof-effect-restrictions}.

\paragraph{The equal-effect example.}
For $F_0=F_1=\operatorname{Bernoulli}(1/2)$, $w=q=1/2$ and $0<\epsilon<1/2$, we have $t_1=h$, $t_0=1-h$, and at $\gamma=0$
\[
|d|\le\{2\max(h,1-h)\}^{-1},\qquad
m_1(v)-m_0(v)=(1-2h)d.
\]
Its maximum absolute value over $h\in[\epsilon,1-\epsilon]$ is
$(1-2\epsilon)/\{2(1-\epsilon)\}$. At either endpoint the corresponding
binary component means attain it, proving sharpness. More generally, under homogeneous
assignment, $h=q$, the population contrast is $\mu_1-\mu_0$ and the
principal contrast is $\mu_1-\mu_0+(v-w)z$, with exact range
\[
z\in[\ell_1(w)-u_0(w),u_1(w)-\ell_0(w)]\cap[-\gamma,\gamma].
\]
Thus combining homogeneous assignment with $\gamma=0$ identifies the
principal contrast without identifying its two principal outcome means.

\paragraph{Calculation for Example~\ref{ex:positive-budget-interior}.}
At the principal upper optimum, the active vertex has
$\delta_0=-\mu_0/(1-t_0)$ and $\delta_1=\delta_0+\gamma$.
Its value and stationary assignment are
\[
 B(h)=0.66+\frac{0.17(h-0.57)}{0.57h}
             +0.4\left(1-\frac{0.43h}{0.57}\right),
 \qquad h_\star=\sqrt{969/1720}=0.750581.
\]
Figure~\ref{fig:positive-budget-objectives} shows the optimized population
and principal upper objectives over the same feasible assignment interval.
The finite-candidate calculation gives the unique principal upper value
$B(h_\star)=0.905263$, whereas the best assignment-boundary value is
$0.895849$.  The stationary root is simple, with $B''(h_\star)=-0.804055$,
and the minimum nondefining feasibility slack exceeds $0.094736$.
The simple stationary root and positive inactive-constraint slack ensure that the unique interior optimum persists under small changes in the observed probabilities and means.

\subsection{Binary population bounds with unspecified coupling}\label{supp:unspecified-computation}

Optimizing over the response coupling separates information supplied by the observed margins from that supplied by a chosen association. For initial arm $j=1$, write $r=e_{11}(\X)$, suppress $\X$ and let
\[
\mathcal R=[\max(0,p_0+p_1-1),\min(p_0,p_1)],\qquad
w_1(r)=\frac r{p_1},\quad w_0(r)=\frac{p_0-r}{1-p_1}.
\]
The same $r$ must serve both response cells, omitting zero-probability cells. Define the conditional SPI contrast $d_g(\X)=\sum_s p_1^s(\X)\{\mu_{(1,1)}(s,\X)-\mu_{(1,0)}(s,\X)\}$. Let $J_s^\pm(r)$ be $J_\epsilon^\pm$ from Proposition~\ref{prop:shared-tilt}, evaluated at $w_s(r)$ and that cell's observed probabilities and means. Its conditional limits are $d_g-\sum_s p_1^sJ_s^-(r)$ and $d_g+\sum_s p_1^sJ_s^+(r)$. Maximize each weighted deviation over $\mathcal R$, then average the conditional limits over $\X$.

For $0<\epsilon<1/2$, this optimization is finite. At $\epsilon=1/2$, compatibility forces all observed and latent treatment probabilities to equal $1/2$, so the population bounds collapse to the g-formula contrast and no coupling optimization is needed. At an assignment endpoint some component has propensity $b\in\{\epsilon,1-\epsilon\}$. If its prevalence is $u\in\{w_s,1-w_s\}$, its feasibility range is
\[
0\le u\le u_{\max}(b),\qquad
u_{\max}(\epsilon)=\frac{1-\epsilon-q}{1-2\epsilon},\quad
u_{\max}(1-\epsilon)=\frac{q-\epsilon}{1-2\epsilon}.
\]
For outcome arm $a$, put $b_a=b^a(1-b)^{1-a}$, $q_a=q^a(1-q)^{1-a}$, $c=b_a/q_a$ and $t=cu$. Partition $\mathcal R$ at its endpoints, assignment-feasibility boundaries and feasible solutions of $cu\in\{\mu_a,1-\mu_a\}$ for both cells, arms and component labels. Evaluating these candidates gives the extrema over all measurable couplings. Each fixed-coupling interval contains the same population SPI functional, so their union is an interval. The proof is in Section~\ref{supp:proof-unspecified-coupling}. A scalar odds-ratio grid, or the two coupling endpoints alone, need not recover it. Principal ratios have coupling-dependent masses and require a separate optimization.

Sufficiency of assignment endpoints at a fixed coupling does not justify restricting the coupling to its endpoints. For example, take $p_0=7/10$, $p_1=2/5$, $q_0=q_1=1/2$, all $\mu_a=1/2$ and $\epsilon=1/10$. The admissible coupling interval is $[1/10,2/5]$. Its upper limits at these boundaries are $24/209$ and $4/15$, whereas $r=4/15$ attains $1874/5355$, exceeding the larger by approximately $8.33$ percentage points. This witness disproves the shortcut without asserting that the exhibited interior point is the maximizer.
\section{Point identification and response sufficiency}\label{supp:point}
The potential outcomes are integrable throughout this section.
\subsection{Point-identification restrictions}\label{supp:point-identification-cases}
For the gaps in \eqref{eq:deltarho}, $t_a-w=(-1)^{1-a}w(1-w)\rho/q_a$. Hence $\delta_{(j,a)}\rho=0$ removes bias in every cell without relying on cancellation. This is an identification condition, distinct from estimator double robustness.

Assumption~\ref{ass:spi} sets $\delta_{(j,a)}=0$ in every mixed cell. The alternative restriction of latent propensity homogeneity (LPH) requires
\begin{equation}\label{eq:lph}
 h_u(j,\X)=h_{u'}(j,\X),\qquad u,u'\in\calU(j,s)\text{ with }e_ue_{u'}>0.
 \tag{LPH}
\end{equation}
This is the assignment-side restriction considered by \citet{ricciardi2020}. 

\begin{corollary}[Population identification]\label{thm:union}
Under Assumptions~\ref{ass:sutva}--\ref{ass:positivity}, the condition
\begin{equation}\label{eq:union}
 \delta_{(j,a)}(j,s,\X)\rho(j,s,\X)=0
\end{equation}
in every mixed cell identifies $\psi(j,a)=\psi_g(j,a)$. Either SPI or LPH is sufficient for all sequence means. A response coupling need not be specified.
\end{corollary}
The proof is in Section~\ref{supp:proof-population-identification}.
Within LSI, this product condition is exactly the corresponding cellwise mean exchangeability. It can hold while distributional SI fails. Section~\ref{app:relations} records the precise relations and counterexamples.

\paragraph{Principal effects.}\label{sec:pce}
\begin{corollary}[Principal identification]\label{thm:pce}
Under Assumptions~\ref{ass:sutva}, \ref{ass:lsi}, \ref{ass:positivity}, \ref{ass:theta} and \ref{ass:spi}, the specified coupling and $\E(e_u)>0$ identify $\psi_u(j,a)$ through \eqref{eq:pce-id}. Contrasts subtract the corresponding principal means, using the response coordinate under each initial treatment.
\end{corollary}
The proof is in Section~\ref{supp:proof-principal-identification}.
LPH alone does not identify these means. For a second-stage contrast, however, LPH together with $\delta_1=\delta_0$ gives
\[
m_{(j,1)}(u,\X)-m_{(j,0)}(u,\X)
=\mu_{(j,1)}(u_j,\X)-\mu_{(j,0)}(u_j,\X).
\]
Principal weighting identifies this contrast even when its two component means remain unidentified, by \eqref{eq:unified-effect}.

\subsection{The specified coupling and principal weighting}\label{supp:plackett}
\begin{lemma}[Principal scores]\label{lem:plackett}
Response margins $p_0,p_1\in(0,1)$ and a specified conditional odds ratio $\theta\in[0,\infty]$ uniquely determine $e_u$. For fixed $0<\theta<\infty$, the scores are continuously differentiable in $(p_0,p_1)$. At $\theta=0$ or $\infty$, they are piecewise affine in the response margins.
\end{lemma}
The proof is in Section~\ref{supp:proof-plackett}. The formulas and margin derivatives $\dot e_{u,j}=\partial e_u/\partial p_j$ are given below.

For $p_0,p_1\in(0,1)$ let $B=1+(\theta-1)(p_0+p_1)$. The Plackett map \citep{plackett1965} is
\begin{equation}\label{eq:plackett}
 e_{11}=
 \begin{cases}
 (B-\sqrt\Delta)/\{2(\theta-1)\},&0<\theta<\infty,\ \theta\ne1,\\
 p_0p_1,&\theta=1,\\
 \max(0,p_0+p_1-1),&\theta=0,\\
 \min(p_0,p_1),&\theta=\infty,
 \end{cases}
 \qquad \Delta=B^2-4\theta(\theta-1)p_0p_1>0
\end{equation}
for $0<\theta<\infty$, with $e_{01}=p_1-e_{11}$, $e_{10}=p_0-e_{11}$ and $e_{00}=1-p_0-p_1+e_{11}$. Its finite-coupling derivatives are
\begin{equation}\label{eq:plackett-partials}
 d_0:=\frac{\partial e_{11}}{\partial p_0}=\frac{\theta p_1-(\theta-1)e_{11}}{\sqrt\Delta},\qquad
 d_1:=\frac{\partial e_{11}}{\partial p_1}=\frac{\theta p_0-(\theta-1)e_{11}}{\sqrt\Delta},
\end{equation}
read as $(d_0,d_1)=(p_1,p_0)$ at $\theta=1$. The remaining derivatives are
\begin{equation}\label{eq:edot}
\begin{array}{c|rrrr}
u&(1,1)&(0,1)&(1,0)&(0,0)\\\hline
\dot e_{u,0}&d_0&-d_0&1-d_0&d_0-1\\
\dot e_{u,1}&d_1&1-d_1&-d_1&d_1-1
\end{array}
\end{equation}
and sum to zero in each row. At $\theta=\infty$, $(d_0,d_1)=(\one\{p_0<p_1\},\one\{p_1<p_0\})$. At $\theta=0$, both equal $\one\{p_0+p_1>1\}$. The respective branch changes occur at $p_0=p_1$ and $p_0+p_1=1$. Our inference results require the corresponding distance $|p_1-p_0|$ or $|p_0+p_1-1|$ to be uniformly bounded away from zero, so small estimation errors do not cross these nondifferentiable points, as required by (R1) in Section~\ref{app:mr}. Identification does not require this separation.

Under SPI, the principal contrast for sequences sharing $z_1$ weights the observed conditional effect by $e_u=p_{z_1}^sw_u$, whereas the contrast given $S(z_1)=s$ weights it by $p_{z_1}^s$, with each weight normalized by its expectation. Constant conditional effects or constant $w_u$ make the two coincide. At $\theta=\infty$ and $p_1\ge p_0$, the scores reduce to the monotone case, with pure observed cells $(1,0)$ and $(0,1)$.

\subsection{Relations among the identifying assumptions}\label{app:relations}
Under LSI, Bayes' rule gives
\begin{equation}\label{eq:mean-si}
\begin{split}
 &\E\{Y(\zbar')\mid\X,Z_1=j,S=s,Z_2=1\}
 -\E\{Y(\zbar')\mid\X,Z_1=j,S=s,Z_2=0\}\\
 &\hspace{35mm}=\frac{w(1-w)\rho\delta_{\zbar'}}{q_1(1-q_1)}.
\end{split}
\end{equation}
Here $\delta_{\zbar'}$ compares the two types in the conditioning cell, even for a cross-arm target. Indeed,
$\widetilde w_1^{\,1}-\widetilde w_1^{\,0}=w(1-w)\rho/\{q_1(1-q_1)\}$, and the other weight changes by its negative.
\begin{proposition}[Implications and their limits]\label{prop:relations}
Under consistency and positivity, LPH together with LSI implies SI. Within LSI, zero bias in a mixed cell is equivalent to $\delta_{\zbar}\rho=0$ and to mean exchangeability of its target outcome in the later treatment. SPI then gives same-initial-arm mean exchangeability, but need not give cross-arm mean exchangeability or distributional SI. LPH and SPI are nonnested, as are SI and LSI. SI need not imply LPH, and nonzero cell biases may cancel after averaging. Apart from stipulated positivity, LSI, SPI and their conjunction impose no observed-law restriction. Neither LSI alone nor SPI alone identifies sequence means.
\end{proposition}
The proof is in Section~\ref{supp:proof-relations}.
The following examples and counterexamples, labeled L1--L5, illustrate the relationships among LSI, SPI, LPH and SI. They support the nonimplication and nonidentification claims in Proposition~\ref{prop:relations}.

Unless specified otherwise, $\X$ is degenerate, $U$ is uniform on $\{0,1\}^2$, and $Z_1$ is an independent Bernoulli$(1/2)$ variable. Thus $w=p_j^s=1/2$. In L1, L2 and L4, later treatment is generated independently of the outcome vector given $(Z_1,U)$, with propensity $h_U(Z_1)$, so LSI holds. L3 uses outcome-dependent assignment instead. All propensities below are interior.

In L1, $h_u=1/2$ and each potential outcome given $U=u$ is Bernoulli with mean $0.2+0.6\one\{u=11\}$. LPH holds, but the outcome-mean gap $\delta=0.6$ in cell $(1,1)$ violates SPI.

In L2, $h_u(j)=0.2+0.6u_{1-j}$ and the outcome vector is independent of $U$, with independent Bernoulli$(1/2)$ coordinates. SPI and SI hold, but the latent propensity difference $\rho=0.6$ violates LPH.

In L3, an auxiliary $Q\sim\mathrm{Bernoulli}(1/2)$ is independent of $(U,Z_1)$, every $Y(\zbar)=Q$, and $\Prob(Z_2=1\mid Q,U,Z_1)=(1+Q)/4$. All principal outcome means equal $1/2$, so SPI holds, but LSI fails. The g-formula gives $\psi_g(j,1)=2/3$ and $\psi_g(j,0)=2/5$, whereas both sequence means equal $1/2$.

In L4, set $h_{u^{(v)}}=\{1+(2v-1)\rho\}/2$ in selected cells, with $0<|\rho|<1$, and $h=1/2$ elsewhere. Binary component means with gap $\delta$ give $q_a=1/2$ and cell bias $(-1)^{1-a}\delta\rho/2$. A single biased cell gives nonidentification, while opposite gaps across response cells can cancel after averaging.

L5 reproduces an arbitrary observed law. Draw $\X$ from that law, $U\mid\X$ from any compatible response coupling, and each $Y(j,a)\mid(\X,U=u)$ from its observed law given $(\X,j,u_j,a)$. Draw $Z_1$ independently of the potential vector given $\X$ at its observed propensity. Set $S=U_{Z_1}$ and draw $Z_2$ independently of the outcomes given $(\X,U,Z_1)$ with $h_u(j,\X)=q_1(j,u_j,\X)$. The identical within-cell component laws and propensities reproduce the observed law and satisfy LSI, SPI and LPH. Pairing L5 with L4 gives identical observed laws with different sequence means under LSI; pairing it with L3 does so under SPI alone.

\begin{remark}[Mean and distributional independence]\label{rem:nonnested}
Two further witnesses separate SI from LSI. First use L2's assignment and, conditional on $U=u$, draw each $Y(j,a)$ independently as Bernoulli$(1/2)$ if $u_{1-j}=1$ and the constant $1/2$ otherwise. LSI and SPI hold, yet in cell $(j,s)$ the outcome atom at $1/2$ has probability $0.2$ given $Z_2=1$ and $0.8$ given $Z_2=0$. Thus distributional SI fails. Also, replacing the outcomes by $Y(j,a)=U_j$ preserves SPI, but in cell $(1,1)$ the cross-arm target has $\delta_{(0,a)}=1$, $w=q_1=1/2$ and $\rho=0.6$, so cross-arm mean exchangeability fails by \eqref{eq:mean-si}.

Conversely let every $Y(\zbar)=Y$, with $Y\mid U=u$ Bernoulli of mean $0.2+0.6\one\{u=11\}$. In cell $(Z_1,S)=(1,1)$, give strata $11$ and $01$ the outcome-dependent propensities $0.8-0.6Y$ and $0.2+0.6Y$ respectively, and use $1/2$ elsewhere. Bayes' rule gives $\Prob(U=11\mid Z_1=1,S=1,Y=y)=0.2+0.6y$. Thus the observed treatment probability is $0.8(0.2)+0.2(0.8)=0.32$ at either $y$, and is $1/2$ in other cells. SI holds, while LSI fails within both mixed strata.
\end{remark}

\subsection{Residual selection and compatible component laws}\label{app:residual}
Retain consistency, initial-treatment ignorability, positivity and integrable potential outcomes, but drop second-stage LSI. Define
\begin{equation}\label{eq:residual-selection}
 r_{\zbar}(u,\X)=\E\{Y(\zbar)\mid\X,Z_1=z_1,U=u,Z_2=z_2\}-m_{\zbar}(u,\X).
\end{equation}
For an analyst-specified common allowance $\eta\ge0$, we impose
\[
|r_{\zbar}(u,\X)|\le\eta
\]
almost surely for every sequence and profile with $e_u(\X)>0$. Measured in outcome units, $\eta$ bounds how much the potential-outcome mean among recipients of the later treatment can differ from the whole profile's mean at the same covariates. LSI makes this difference zero.

\begin{proposition}[Residual selection decomposition]\label{prop:residual-selection}
Under these retained assumptions,
\begin{equation}\label{eq:residual-bias}
 \psi_g(\zbar)-\psi(\zbar)=B_{\zbar}+R_{\zbar},\qquad
 R_{\zbar}=\E\sum_s p_{z_1}^s\sum_u\widetilde w_u^{z_2}r_{\zbar}(u,\X),
\end{equation}
where $B_{\zbar}$ uses the actual latent means in \eqref{eq:bias}. If $|r_{\zbar}|\le\eta_{\zbar}$, then $|R_{\zbar}|\le\eta_{\zbar}$ and the additional contrast bias is at most $\eta_{\zbar}+\eta_{\zbar'}$.
\end{proposition}
The proof is in Section~\ref{supp:proof-residual-selection}.
Under a common allowance $\eta$, the residual contribution to population sequence-mean bias is at most $\eta$ in absolute value, and at most $2\eta$ for a contrast. These allowances need not be sharp. The following result imposes the additional conditions needed for compatibility with the full observed outcome distributions and stipulated support.

\begin{proposition}[Compatibility with residual selection]\label{prop:residual-compatibility}
Suppose every potential outcome $Y(\zbar)$ is supported on $[a,b]$. The fixed observed law has $Y\in[a,b]$ and first-treatment probabilities bounded away from zero and one. The measurable response coupling reproduces its margins, and the measurable latent assignment probabilities are bounded away from zero and one and satisfy $q_1=\sum_u w_uh_u$ in every supported cell. The inputs $(\delta_{\zbar},r_{\zbar})$ are measurable, and first-stage ignorability is retained without second-stage LSI. In a mixed cell let $k_v=h_{u^{(v)}}^{(z_2)}$, $t=wk_1/q_{z_2}$, $\mu=\mu_{\zbar}(s,\X)$ and $F$ be the observed conditional outcome law. Put
\begin{equation}\label{eq:residual-components}
\begin{split}
n_1&=\mu+(1-t)(\delta_{\zbar}+r_1-r_0),\qquad
n_0=\mu-t(\delta_{\zbar}+r_1-r_0),\\
m_v&=n_v-r_v,\qquad a_v^{\rm mis}=n_v-\frac{r_v}{1-k_v}.
\end{split}
\end{equation}
A compatible full law exists if and only if, almost surely in every positive-probability cell and for each sequence,
\begin{equation}\label{eq:residual-compatible}
T_F^-(t)\le n_1\le T_F^+(t),\qquad a_v^{\rm mis}\in[a,b]\quad(v=0,1),
\end{equation}
where $T_F^\pm$ are defined in \eqref{eq:trimmed-means}. In a pure cell set the sole $n_v=\mu$ and retain only its missing-mean restriction. For real-valued outcomes without a support restriction, the trimmed-mean condition remains sufficient when the constructed potential outcomes are integrable. An integrable observed outcome and uniformly bounded residuals suffice under uniform positivity.
\end{proposition}
The proof is in Section~\ref{supp:proof-residual-compatibility}.

For fixed $(\delta,\rho,\theta)$ in a nonempty compatible residual class, let $c=\psi_g(\zbar)-\psi_g(\zbar')-(B_{\zbar}-B_{\zbar'})$. The allowance $|r|\le\eta$ gives the outer interval $[c-2\eta,c+2\eta]$. This interval requires the compatibility conditions in Proposition~\ref{prop:residual-compatibility} and need not be sharp.

\section{Influence functions, remainders and inference}\label{supp:inference-statements}\label{app:estimation}

\subsection{Common residuals and ratio estimation}\label{supp:distribution-scores}\label{supp:eif}
Write $q_{j,s}=\Prob(Z_2=1\mid\X,Z_1=j,S=s)$ and $\mu_{j,s,a}=\E(Y\mid\X,j,s,a)$. The vector $\nu$ contains the two response margins, four $q_{j,s}$ and eight outcome means. Its conditional residuals are
\begin{equation}\label{eq:distribution-base-scores}
\begin{aligned}
S_{p_j}=A_j&=\frac{\one\{Z_1=j\}}{\pi_j}(S-p_j),&
S_{q_{j,s}}&=\frac{\one\{Z_1=j,S=s\}}{\pi_jp_j^s}(Z_2-q_{j,s}),\\
S_{\mu_{j,s,a}}=B_{(j,a),s}&=W_{j,s,a}(Y-\mu_{j,s,a}),&
W_{j,s,a}&=\frac{\one\{Z_1=j,S=s,Z_2=a\}}
{\pi_jp_j^s q_{j,s}^{a}(1-q_{j,s})^{1-a}}.
\end{aligned}
\end{equation}
All probabilities condition on $\X$. A conditional atom probability on a known finite support has residual $W_{j,s,a}\{\one(Y=y)-\Prob(Y=y\mid\X,j,s,a)\}$. Nuisance $L_r$ norms are over $\Prob_X$, maximized over discrete arguments. The norm of $\Gamma=(\pi,p,q,\mu)$ is the maximum of its four coordinate norms, and $\|\cdot\|_\infty$ is the corresponding essential supremum.

For a dominated submodel with score $\ell$ and a measurable function $g=g(\X)$, conditional differentiation gives
\begin{equation}\label{eq:score-conv}
\E\{g\,\partial_\eta\nu_l\}=\E(gS_l\ell),\qquad
\E\{g\,\E(\ell\mid\X)\}=\E[\{g-\E(g)\}\ell].
\end{equation}
The identities hold first for bounded $g$ and then by truncation when the products are integrable. For a smooth conditional quantity $G(\nu)$ put
\begin{equation}\label{eq:common-contribution}
\mathcal C_G=G+\nabla G^\top S_\nu.
\end{equation}
Thus $\mathcal C_G-\E G$ is the gradient of $\E G$. For $T=\E A/\E m$, $D=\E m>0$, the quotient rule gives
\begin{equation}\label{eq:common-ratio}
\phi_T=\frac{\mathcal C_A-T\mathcal C_m}{D},\qquad
\widehat T=\frac{\Pn\widehat{\mathcal C}_A}{\Pn\widehat{\mathcal C}_m}.
\end{equation}
This retains numerator--mass covariance. The gradient is efficient when the maintained model has tangent-space closure $L_0^2(\Prob)$, justified in the corresponding proofs in Section~\ref{supp:proofs}.

For independent identically distributed observations, cross-fitting uses fixed $K$ folds with $n_k/n$ bounded away from zero, fits outside each evaluation fold and averages contributions with weights $n_k/n$. Folds may be assigned independently of the observations or uniformly within initial-treatment strata, with each fold preserving the sample stratum proportions up to rounding. Section~\ref{supp:common-crossfit} gives the common expansion under $L_2$ score consistency and negligible drift. The following results supply these conditions.

\subsection{SPI functionals and inference}\label{app:eif-pce}\label{app:mr}
Let $s^*=u_j$, $\mu=\mu_{(j,a)}(s^*,\X)$ and
$W_u=e_u+\sum_b\dot e_{u,b}A_b$. The SPI ratio estimator is
\begin{equation}\label{eq:mr-estimator}
\widehat\psi_u^{\rm aug}(j,a)
=\frac{\Pn\{\widehat\mu_{(j,a)}(s^*,\X)\widehat W_u+\widehat e_u\widehat B_{(j,a),s^*}\}}
{\Pn\widehat W_u}.
\end{equation}

The estimation statements use the following additional conditions.
\begin{enumerate}[label=(R\arabic*),nosep]
\item True and fitted observed propensities, response margins and action probabilities belong to $[\epsilon',1-\epsilon']$. Margin pairs lie in a compact convex $\calK\subset(0,1)^2$ on which the first two derivatives of $e_u$ are uniformly bounded. At $\theta=\infty$ require $\inf_{\calK}|p_1-p_0|>0$, and at $\theta=0$ require $\inf_{\calK}|p_0+p_1-1|>0$.
\item $\E(Y^4)<\infty$, and true and fitted outcome regressions are uniformly bounded.
\end{enumerate}

\begin{theorem}[SPI influence function and inference]\label{thm:eif-pce}\label{thm:dml}
Under Assumptions~\ref{ass:sutva}, \ref{ass:lsi}, \ref{ass:positivity}, \ref{ass:theta} and \ref{ass:spi}, suppose $\E(Y^2)<\infty$, $\epsilon_p\le p_j\le1-\epsilon_p$, and either $\theta$ lies in a compact subinterval of $(0,\infty)$ or $\theta\in\{0,\infty\}$ with the separation conditions on response margins in Section~\ref{supp:plackett}. For $\E e_u>0$, the efficient influence function is
\begin{equation}\label{eq:eif-pce}
\varphi_{u,(j,a)}
=\frac{(\mu-\psi_u(j,a))W_u+e_u B_{(j,a),s^*}}{\E e_u}.
\end{equation}
Principal contrasts use the difference of the corresponding influence functions, evaluating each outcome residual at its own slot $u_j$.

Under the additional conditions (R1)--(R2), suppose the common cross-fitting scheme uses $L_2$-consistent nuisances satisfying, in each fold,
\begin{equation}\label{eq:dml-rates}
\|\widehat p-p\|_2\{\|\widehat p-p\|_2+\|\widehat\pi-\pi\|_2\}
+\|\widehat\mu-\mu\|_2\{\|\widehat\pi-\pi\|_2+\|\widehat p-p\|_2+\|\widehat q-q\|_2\}
=o_{\Prob}(n^{-1/2}).
\end{equation}
The estimator in \eqref{eq:mr-estimator} averages the fold-weighted numerator and mass contributions before division and has influence function \eqref{eq:eif-pce} and an efficient normal limit. The expansion and empirical covariance consistency hold jointly for finitely many principal means and their contrasts.
\end{theorem}
The proof and exact drift are in Section~\ref{supp:proof-dml}. The product condition does not require uniform nuisance convergence.

For the sequence functional, set $\nu_{\zbar}=\sum_s p_j^s\mu_{(j,a)}(s,\X)$. The same calculus gives the classical longitudinal gradient \citep{bang2005,luedtke2017}
\begin{equation}\label{eq:eif-ate}
\varphi^g_{\zbar}
=\nu_{\zbar}-\psi_g(\zbar)
+\{\mu_{(j,a)}(1,\X)-\mu_{(j,a)}(0,\X)\}A_j
+\sum_s p_j^s B_{(j,a),s}.
\end{equation}
The last sum cancels the response-margin denominator in \eqref{eq:distribution-base-scores}. Thus observed positivity and $\E(Y^2)<\infty$ suffice for its square integrability without a lower bound on $p_j^s$. Under Corollary~\ref{thm:union}, $\psi_g=\psi$ and this is also the causal sequence-mean EIF.

\begin{corollary}[Sequence-mean inference]\label{cor:seq-dml}
Under Corollary~\ref{thm:union}, (R2), and the observed probability bounds in (R1), use the common fixed-fold scheme. The estimator solving $\Pn\widehat\varphi^g_{\zbar}=0$ is efficient if its nuisances are $L_2$ consistent and, foldwise,
\[
\|\widehat p-p\|_2\|\widehat\pi-\pi\|_2
+\|\widehat\mu-\mu\|_2(\|\widehat\pi-\pi\|_2+\|\widehat q-q\|_2)
=o_{\Prob}(n^{-1/2}).
\]
The result holds jointly across sequences and contrasts.
\end{corollary}
The proof is in Section~\ref{supp:proof-sequence-inference}.

\subsection{Compatible full-outcome endpoints}\label{supp:distribution-regularity}\label{supp:distribution-inference}
At fixed coupling, the population sequence and shared-assignment contrast bounds without additional outcome restrictions have finitely many candidates from the two assignment endpoints and lower or upper tail allocations. On a smooth active branch write an endpoint as $T=\E\{K(\nu,\bm R)\}$. Here $\nu$ contains the response margins, observed second-treatment probabilities and outcome means. For the conditional law $F_r=F_{j_ra_r,s_r,\X}$, each $R_r=R^{\sigma_r}_{F_r}\{t_r(\nu)\}$ is a lower or upper tail integral, $R^-_F=L_F$ or $R^+_F=U_F$, and $K$ includes the outer response weights. Other population endpoints under the local restrictions of Theorem~\ref{thm:contraction} are covered only when such a finite representation is established and the conditions below hold. The contraction theorem alone does not supply that representation.

At cutoffs $Q^-=F^{-1}(t)$ and $Q^+=F^{-1}(1-t)$,
\begin{equation}\label{eq:tail-gradients}
\begin{aligned}
\xi^-_F(Y;t)&=tQ^--(Q^--Y)_+-L_F(t),
&\partial_tL_F(t)&=Q^-,\\
\xi^+_F(Y;t)&=tQ^++(Y-Q^+)_+-U_F(t),
&\partial_tU_F(t)&=Q^+.
\end{aligned}
\end{equation}
The centered hinge terms correct estimation of the conditional outcome law, while the cutoff derivatives correct estimation of the trimming fraction. With $K_r=\partial K/\partial R_r$, the combined derivative in a base coordinate is
\begin{equation}\label{eq:distribution-total-derivative}
D_l=\frac{\partial K}{\partial\nu_l}
+\sum_rK_rQ_r\frac{\partial t_r}{\partial\nu_l}.
\end{equation}
The first derivative holds the tail arguments fixed. Let $S_l$ be the inverse-probability residual for $\nu_l$ and $W_r$ its outcome history weight, defined in \eqref{eq:distribution-base-scores}.

Impose the following sufficient conditions on common neighbourhoods, foldwise.
\begin{enumerate}[label=(D\arabic*),nosep]
\item\label{cond:distribution-domain}
Outcomes have bounded stipulated support, containing all fitted means and cutoffs. True and fitted base coordinates and tail integrals lie in a compact domain with residual denominators bounded below and tilts in $[c,1-c]$. True observed treatment probabilities are uniformly inside $(\epsilon,1-\epsilon)$. Candidates and tilts have bounded first and second derivatives. Pure cells retain their status and contribute only observed-law means.
\item\label{cond:distribution-branches}
Assignment branches, cutoff-feasibility comparisons and active candidates have uniform positive gaps after algebraically identical candidates are identified. Base coordinates, tail integrals and cutoffs converge uniformly, preserving feasible candidates and selected branches with probability tending to one.
\item\label{cond:distribution-cutoffs}
Each cutoff has either a conditional density bounded above and away from zero on a common neighbourhood, or a known finite outcome support with its trimming probability a fixed positive distance inside the cutoff atom's jump. In the latter case, fitted cutoffs are exact quantiles on that same support and fitted atom probabilities converge uniformly. Hence cutoff atoms eventually agree. Moving atom locations, numerical cutoffs beside atoms and jump endpoints are excluded.
\item\label{cond:distribution-rates}
With a fixed number of nonvanishing cross-fitting folds, define
\[
r_n=\|\widehat\nu-\nu\|_2+
\sum_r(\|\widehat R_r-R_r\|_2+\|\widehat Q_r-Q_r\|_2),
\qquad s_n=\|\widehat\pi-\pi\|_2.
\]
True tails and cutoffs use true tilts, fitted quantities use fitted tilts. Require $r_n+s_n=o_{\Prob}(1)$ and
\begin{equation}\label{eq:distribution-rates}
r_n^2+s_nr_n=o_{\Prob}(n^{-1/2}).
\end{equation}
\end{enumerate}

\begin{theorem}[One-step inference for population endpoints]\label{thm:distribution-inference}
Fix the response coupling and latent positivity margin. For a population sequence bound or shared-assignment contrast endpoint with the preceding finite representation, impose conditions (D1)--(D4) of Section~\ref{supp:distribution-regularity}.
The efficient influence function on the stipulated outcome support is
\begin{equation}\label{eq:distribution-eif}
\phi_T=K-T+\sum_lD_lS_l+\sum_rK_rW_r\xi_r.
\end{equation}
The estimator averaging the fitted $K+\sum_lD_lS_l+\sum_rK_rW_r\xi_r$ satisfies
\[
\sqrt n(\widehat T-T)=n^{-1/2}\sum_{i=1}^n\phi_T(\calO_i)+o_{\Prob}(1).
\]
Its empirical influence variance is consistent. The expansion holds jointly for finitely many endpoints, with a joint normal approximation when their covariance is nondegenerate.
\end{theorem}
The proof is in Section~\ref{supp:proof-distribution-inference}.

\subsection{A fixed grid of restricted binary endpoints}\label{supp:restricted-inference}
For fixed endpoint $k$, let $A_k$ be its conditional numerator and $m_k$ its mass. Population endpoints have $m_k=1$; principal endpoints have $A_k=e_uB_k$ and $m_k=e_u$, where $B_k$ is the conditional principal bound. Use \eqref{eq:common-contribution}--\eqref{eq:common-ratio}, differentiating response weights and all moving candidate coordinates. A simple stationary or feasibility root obeys
\begin{equation}\label{eq:restricted-root-derivative}
P\{h(\nu),\nu\}=0,\qquad
\nabla_\nu h=-P_\nu/P_h,\qquad
\nabla_\nu f\{h(\nu),\nu\}=f_\nu+f_h\nabla_\nu h.
\end{equation}
At a stationary objective the last term vanishes. At a moving feasibility boundary it must generally be retained.

True and fitted nuisance functions lie in a common compact interior domain with neighbourhoods of uniform radius, true $q$ uniformly inside $(\epsilon,1-\epsilon)$, and positive residual and candidate denominators and principal masses bounded below. Active candidates have persistent feasible charts with constant-rank vertex equations, uniformly slack inactive constraints and simple roots whose gradients $-P_\nu/P_h$ are uniformly bounded and Lipschitz on these neighbourhoods. This root condition is invariant to nonzero rescaling of the root equation. Complete candidate gradients are bounded and uniformly Lipschitz. Either algebraically distinct optimizers have uniform positive gaps, or their joined value is continuously differentiable with uniformly Lipschitz gradient throughout the neighbourhood. Uniform nuisance consistency retains these neighbourhoods.

Nonunique optimizers can satisfy these conditions when their complete gradients agree where the branches meet and the joined gradient is Lipschitz. Value ties alone do not suffice, and positive-probability ties with disagreeing derivatives are excluded.

At $\gamma=0$, first impose $\delta_0=\delta_1=d$. The contrast $\mu_1-\mu_0+(t_0-t_1)d$ is independent of the principal label, so Theorem~\ref{thm:contraction} gives assignment-endpoint sufficiency for this conditional contrast too. Marginal weights still differ. This equality reduction removes redundant strip constraints.

\begin{corollary}[Restricted finite-grid inference]\label{cor:restricted-inference}
On the preceding domain, suppose observations are independent, folds are fixed and nonvanishing, and foldwise $r_n=\|\widehat\nu-\nu\|_2$, $s_n=\|\widehat\pi-\pi\|_2$ tend to zero and satisfy \eqref{eq:distribution-rates}. Then
\begin{equation}\label{eq:restricted-ratio-inference}
\widehat T_k=\frac{\Pn\widehat{\mathcal C}_{A_k}}{\Pn\widehat{\mathcal C}_{m_k}},\quad
\phi_k=\frac{\mathcal C_{A_k}-T_k\mathcal C_{m_k}}{\E m_k},\quad
\sqrt n(\widehat{\bm T}-\bm T)=n^{-1/2}\sum_i\bm\phi(\calO_i)+o_{\Prob}(1).
\end{equation}
Empirical influence covariance is consistent. For positive marginal variances, pointwise normal intervals have asymptotic nominal coverage. For $M$ fixed endpoints, the critical value $z_{1-\alpha/(2M)}$ gives simultaneous coverage at least $1-\alpha$ asymptotically.
\end{corollary}
The proof is in Section~\ref{supp:proof-restricted-inference}.

\section{Simulation specifications and evaluation}\label{supp:sims}\label{supp:sim-design}
The four scenarios use $1000$ independent draws at each of $n\in\{2000,8000,32000\}$ and the estimation procedure described in Section~\ref{sec:simulations}.

\paragraph{Generating laws.}
Synthetic covariates $(X_1,X_2)$ are uniform on $\{-1,1\}^2$, with $\Prob(Z_1=1\mid\X)=1/2$. The two potential responses are conditionally independent, so $\theta=1$. In the Baseline and Interior scenarios,
\[
 p_0=0.60,\qquad p_1=0.65,\qquad q_{js}=0.70,\qquad
 \mu_{js0}=0.20,\qquad \mu_{js1}=0.70.
\]
The Covariates scenario replaces these by
\[
\begin{aligned}
 \operatorname{logit}p_0(x)&=\operatorname{logit}(0.60)+0.075x_1+0.05x_2,\\
 \operatorname{logit}p_1(x)&=\operatorname{logit}(0.65)-0.05x_1+0.075x_2,\\
 \operatorname{logit}q_{js}(x)&=\operatorname{logit}(0.70)+0.05x_1-0.05x_2,\\
 \operatorname{logit}\mu_{jsa}(x)&=\operatorname{logit}(0.20+0.50a)+0.05x_1+0.05x_2.
\end{aligned}
\]
The History scenario additionally adds $0.05(2j-1)+0.04(2s-1)x_2$ to the later-treatment logit and $0.05(2j-1)(2s-1)+0.05(2a-1)x_2$ to the outcome logit.

The latent positivity margin is $\epsilon=0.08$ in every scenario. Within a mixed response cell, with alternative-response prevalence $w$ and observed later-treatment probability $q$, the assignment probabilities are
\[
 h_v=q+(v-w)r,\qquad
 r=0.2\min\left\{\frac{q-\epsilon}{w},\frac{1-\epsilon-q}{1-w}\right\}.
\]
With $t_1=wh_1/q$ and $t_0=w(1-h_1)/(1-q)$, potential outcomes are Bernoulli with component means
\[
 m_{jsa}(v,x)=\mu_{jsa}(x)+0.04\{v-t_a(x)\},
\]
and independent of assignment given $(U,\X)$. Their selected mixtures reproduce the specified observed means exactly. The component gaps satisfy $\delta_0=\delta_1=0.04$, so the generating law satisfies the zero effect budget while violating SPI. The recorded outcome is $Y(Z_1,Z_2)$. Target bounds optimize over all compatible component means under the stipulated budget, rather than assuming that the generating gap is known.

The budgets are $\gamma=0$ for Baseline, Covariates and History, and $1$ for Interior. True endpoints average the exact conditional bound contributions from Proposition~\ref{prop:finite-effect-budget} over the four covariate profiles, dividing principal numerators by the true stratum mass. Candidate enumeration also verifies feasibility and endpoint derivatives at each generating law.

\paragraph{Common nuisance-fitting procedure.}
Each synthetic profile is encoded in the application's input fields and fitted with the preprocessing and ridge-logistic models in Section~\ref{supp:emp-nuisance}. The simulations assess this procedure on synthetic covariate and outcome distributions rather than reproducing the application's empirical distribution. Five outer folds are stratified by initial treatment, and five inner folds by each model's response. Predictions are the raw probabilities at the penalty minimizing cross-validated binomial deviance.

The initial propensity $\pi$ is estimated by the same ridge-logistic procedure, despite its true value $1/2$. Only the oracle uses that known value. The fixed tuning path alone does not establish consistency or the nuisance rates in Section~\ref{sec:estimation}.

\paragraph{Estimation and inference.}
All scalar confidence intervals use the cross-fitted one-step estimate and its estimated influence standard error,
\[
 [\widehat T-1.96\widehat{\mathrm{se}}(\widehat T),\,
   \widehat T+1.96\widehat{\mathrm{se}}(\widehat T)].
\]
Complete endpoint derivatives include moving candidate coordinates, response weights and the principal numerator/mass ratio. Covariate averaging uses the sample rather than known profile probabilities. The oracle one-step estimator uses true nuisance functions throughout while retaining empirical covariate averaging. The plug-in comparison uses the same estimated law without correction and makes no interval claim.

For a target's identified interval $[T_L,T_U]$, a Bonferroni adjustment of the endpoint intervals gives the outward region
\[
 [\widehat T_L-z_{0.9875}\widehat{\mathrm{se}}_L,\,
   \widehat T_U+z_{0.9875}\widehat{\mathrm{se}}_U].
\]
Under Corollary~\ref{cor:restricted-inference}, this adjustment gives asymptotic simultaneous two-sided endpoint coverage of at least $95\%$. Whole-set enclosure needs only one tail per endpoint, so the same critical value conservatively gives asymptotic enclosure of at least $97.5\%$. Its length equals the estimated span $\widehat T_U-\widehat T_L$ plus sampling padding $z_{0.9875}(\widehat{\mathrm{se}}_L+\widehat{\mathrm{se}}_U)$. Containment of the generating causal effect remains a distinct, weaker event.

\begingroup
\small
\setlength{\tabcolsep}{4pt}
\begin{table}[t]
\centering
\caption{Point-estimation comparisons for every Baseline and Interior endpoint at $n=2000$, with $R=1000$ generated datasets per scenario. Plug-in and One-step use the same fitted nuisance functions. Oracle uses true nuisance functions throughout. Signed bias and RMSE use finite estimates. No plug-in interval is asserted.}
\label{tab:simulation-method-comparison}
\begin{tabular}{@{}lrrrrrr@{}}
\toprule
Endpoint & \multicolumn{2}{c}{Plug-in} & \multicolumn{2}{c}{One-step} & \multicolumn{2}{c}{Oracle} \\
\cmidrule(lr){2-3}\cmidrule(lr){4-5}\cmidrule(l){6-7}
 & Bias & RMSE & Bias & RMSE & Bias & RMSE \\
\midrule
\multicolumn{7}{@{}l}{\textbf{Baseline}} \\
Pop. L & -0.003 & 0.025 & -0.020 & 0.041 & 0.002 & 0.026 \\
Pop. U & -0.004 & 0.018 & 0.008 & 0.024 & 0.001 & 0.020 \\
Princ. L & 0.010 & 0.033 & -0.024 & 0.056 & 0.002 & 0.031 \\
Princ. U & 0.004 & 0.021 & 0.006 & 0.028 & 0.001 & 0.024 \\
\midrule
\multicolumn{7}{@{}l}{\textbf{Interior}} \\
Pop. L & -0.013 & 0.030 & -0.023 & 0.048 & -0.001 & 0.026 \\
Pop. U & -0.007 & 0.017 & -0.000 & 0.019 & 0.000 & 0.019 \\
Princ. L & 0.021 & 0.042 & -0.041 & 0.090 & -0.000 & 0.023 \\
Princ. U & -0.011 & 0.028 & -0.003 & 0.037 & 0.001 & 0.033 \\
\bottomrule
\end{tabular}
\end{table}
\endgroup

\paragraph{Additional scenarios.}
Table~\ref{tab:simulation-reduced-detail} gives the Covariates and History results at all three sample sizes. Baseline and Interior results appear in Table~\ref{tab:simulation-reduced-main}. The procedures and tuning rules are identical across these scenarios.

In Covariates, sampling error and interval width decrease with sample size, but population lower coverage is not monotone. It is $92.6\%$, $90.6\%$ and $93.2\%$ over the three sizes. At $n=32000$, the known-nuisance interval covers $95.0\%$, compared with $93.2\%$ for the fitted procedure. At this sample size, the two principal intervals cover near or above $95\%$.

History retains more pronounced lower-endpoint deficits despite shrinking bias. Population lower coverage is $90.1\%$, $89.8\%$ and $92.2\%$ over the three sizes. At $n=32000$, principal lower coverage is $92.9\%$. The population and principal lower known-nuisance coverages are $94.0\%$ and $94.5\%$, respectively.

\begingroup

\begingroup
\small
\setlength{\tabcolsep}{3pt}
\begin{longtable}{@{}rlrrrrrr@{}}
\caption{Covariates and History scenarios at all three sample sizes. Every row uses the application nuisance learner and the same one-step normal interval, with $R=1000$ generated datasets per cell. Truth is the identified endpoint.  Width is mean interval length.}\label{tab:simulation-reduced-detail}\\
\toprule
$n$ & Endpoint & Truth & Bias & RMSE & SD & Cov. (\%) & Width \\
\midrule\endfirsthead
\toprule
$n$ & Endpoint & Truth & Bias & RMSE & SD & Cov. (\%) & Width \\
\midrule\endhead
\bottomrule\endfoot
\multicolumn{8}{@{}l}{\textbf{Covariates}} \\*
2,000 & Pop. L & -0.0579 & -0.024 & 0.043 & 0.036 & 92.6 & 0.151 \\*
2,000 & Pop. U & 0.6498 & 0.006 & 0.022 & 0.021 & 91.8 & 0.078 \\*
2,000 & Princ. L & -0.0585 & -0.029 & 0.058 & 0.050 & 95.0 & 0.208 \\*
2,000 & Princ. U & 0.6506 & 0.005 & 0.027 & 0.027 & 92.6 & 0.097 \\*
\addlinespace[2pt]
8,000 & Pop. L & -0.0579 & -0.007 & 0.017 & 0.015 & 90.6 & 0.055 \\*
8,000 & Pop. U & 0.6498 & 0.000 & 0.010 & 0.010 & 94.6 & 0.038 \\*
8,000 & Princ. L & -0.0585 & -0.005 & 0.019 & 0.018 & 93.8 & 0.071 \\*
8,000 & Princ. U & 0.6506 & -0.000 & 0.012 & 0.012 & 94.8 & 0.047 \\*
\addlinespace[2pt]
32,000 & Pop. L & -0.0579 & -0.001 & 0.007 & 0.007 & 93.2 & 0.025 \\*
32,000 & Pop. U & 0.6498 & -0.000 & 0.005 & 0.005 & 95.3 & 0.019 \\*
32,000 & Princ. L & -0.0585 & -0.000 & 0.008 & 0.008 & 95.3 & 0.031 \\*
32,000 & Princ. U & 0.6506 & -0.000 & 0.006 & 0.006 & 96.6 & 0.023 \\
\midrule
\multicolumn{8}{@{}l}{\textbf{History}} \\*
2,000 & Pop. L & -0.0522 & -0.022 & 0.045 & 0.039 & 90.1 & 0.152 \\*
2,000 & Pop. U & 0.6484 & 0.006 & 0.023 & 0.022 & 90.8 & 0.079 \\*
2,000 & Princ. L & -0.0564 & -0.025 & 0.061 & 0.056 & 92.7 & 0.212 \\*
2,000 & Princ. U & 0.6555 & 0.005 & 0.027 & 0.027 & 93.4 & 0.098 \\*
\addlinespace[2pt]
8,000 & Pop. L & -0.0522 & -0.008 & 0.016 & 0.014 & 89.8 & 0.056 \\*
8,000 & Pop. U & 0.6484 & 0.001 & 0.010 & 0.010 & 93.8 & 0.039 \\*
8,000 & Princ. L & -0.0564 & -0.008 & 0.019 & 0.018 & 94.6 & 0.073 \\*
8,000 & Princ. U & 0.6555 & 0.001 & 0.012 & 0.012 & 93.7 & 0.047 \\*
\addlinespace[2pt]
32,000 & Pop. L & -0.0522 & -0.001 & 0.007 & 0.007 & 92.2 & 0.025 \\*
32,000 & Pop. U & 0.6484 & 0.000 & 0.005 & 0.005 & 94.5 & 0.019 \\*
32,000 & Princ. L & -0.0564 & -0.001 & 0.009 & 0.009 & 92.9 & 0.032 \\*
32,000 & Princ. U & 0.6555 & 0.000 & 0.006 & 0.006 & 94.4 & 0.023 \\
\end{longtable}
\endgroup
\endgroup

\paragraph{Paired diagnostic comparison.}\label{supp:diagnostic-replay}
The first $200$ datasets in each Baseline, Interior and History cell at $n=2000,8000,32000$ compare the application learner with saturated conditional Bernoulli estimates $(k+1/2)/(m+1)$, where $k$ successes occur among $m$ training observations in the conditioning cell. The same observations, outer folds, endpoint optimization and inference rules are used. These $1800$ diagnostic datasets form a subset of the study. Table~\ref{tab:simulation-diagnostic-replay} reports the comparison for three selected endpoints.

For each endpoint, the scaled influence remainder is $R_n=\sqrt n\{\widehat T-T_0-\Pn\varphi_{T,0}\}$. The principal true-influence average equals $(\widehat T_{\mathrm{or}}-T_0)\widehat m_{\mathrm{or}}/m_0$, with corrected empirical oracle mass $\widehat m_{\mathrm{or}}$ and true mass $m_0$. Population masses equal one. A separate drift calculation holds each fitted fold fixed, takes exact expectations under the finite generating law and combines them with fold-size weights. For principal targets it uses the ratio of expected corrected numerator and mass, not a conditional expectation of the cross-fitted ratio.

Saturated fitting has larger paired mean squared error in $33$ of $36$ scalar comparisons. At $n=32000$, History population lower coverage is $89.5\%$, $88.5\%$ and $92.5\%$ for application, saturated and oracle procedures on these same $200$ draws. The application--saturated difference is only one percentage point, with paired Monte Carlo standard error $1.42$ points. Mean fixed-fit drift remains negative for this endpoint with either learner. These results do not isolate ridge regularization as the cause of the deficits. 

\begingroup
\small
\setlength{\tabcolsep}{4pt}
\begin{table}[t]
\centering
\caption{Paired diagnostic comparison using the same $200$-draw subset (the first $200$ draws per cell) of the $R=1000$ study. Both learners use the same one-step endpoint procedure. RMS($R_n$) summarizes the scaled true-influence remainder. Coverage uses all $200$ draws. Other metrics use finite estimates, numbering $198$ and $196$ for the saturated Baseline and History rows at $n=2000$, and $200$ for every other displayed row.}
\label{tab:simulation-diagnostic-replay}
\begin{tabular}{@{}lrlrrr@{}}
\toprule
Case & $n$ & Learner & RMSE & Cov. (\%) & RMS($R_n$) \\
\midrule
Baseline pop. U & 2,000 & Application & 0.024 & 91.0 & 0.493 \\
Baseline pop. U & 2,000 & Saturated & 0.027 & 85.5 & 0.754 \\
Baseline pop. U & 32,000 & Application & 0.005 & 95.5 & 0.041 \\
Baseline pop. U & 32,000 & Saturated & 0.005 & 95.0 & 0.048 \\
\midrule
Interior princ. L & 2,000 & Application & 0.092 & 94.0 & 3.889 \\
Interior princ. L & 2,000 & Saturated & 0.114 & 88.0 & 5.015 \\
Interior princ. L & 32,000 & Application & 0.005 & 96.0 & 0.209 \\
Interior princ. L & 32,000 & Saturated & 0.007 & 96.5 & 0.920 \\
\midrule
History pop. L & 2,000 & Application & 0.044 & 91.5 & 1.750 \\
History pop. L & 2,000 & Saturated & 0.083 & 81.5 & 3.449 \\
History pop. L & 32,000 & Application & 0.007 & 89.5 & 0.526 \\
History pop. L & 32,000 & Saturated & 0.009 & 88.5 & 1.080 \\
\bottomrule
\end{tabular}
\end{table}
\endgroup

\section{Gym application definitions and implementation}\label{supp:empirical}

\subsection{Cohort and recorded sequence}\label{supp:gym-design}

The study cohort contains $1{,}932$ members, $963$ and $969$ in the respective arms. All $79$ intended/actual assignment mismatches remain included. The initial intervention is the intended assignment under the study's implementation process. The original protocol specifies equal allocation, but intended codes are available only for retained members. First-stage exchangeability is assumed in this population, with its propensity estimated from baseline covariates, since retention can induce selection.
The week-1, week-2 and week-5 attendance indicators and targets are defined in Section~\ref{sec:application}. The eight observed $(Z_1,S,Z_2)$ histories have $318$, $68$, $137$, $446$, $236$, $58$, $149$ and $520$ members in lexicographic order.

\subsection{Attendance interventions and identifying assumptions}\label{supp:emp-policy}

All quantities refer to the retained population. Under initial assignment $j$, intervention $a=0$ prescribes no week-2 visit to the fitness chain, whereas $a=1$ prescribes at least one visit. The initial program continues, including attendance-dependent rewards and subsequent unrestricted behavior. Assumption~\ref{ass:sutva} is imposed on this binary intervention. For each member and fixed initial assignment, all positive week-2 visit patterns are assumed to produce the same potential week-5 outcome $Y(j,1)$. This restriction excludes effects of visit number or timing beyond the attendance indicator and is not established by the weekly records. Assumption~\ref{ass:sutva} also requires no interference and consistency of the recorded first-week response.

With first-stage exchangeability in the retained population, LSI and positivity as in Assumptions~\ref{ass:lsi}--\ref{ass:positivity}, the binary mixture representation and sharp bounds apply. The effect budget restricts latent effect differences under LSI without relaxing consistency or retention exchangeability. Residual selection within response profiles instead requires the compatible selected and missing component means in Section~\ref{app:residual}.

\subsection{Nuisance estimation and compatibility}\label{supp:emp-nuisance}\label{supp:emp-reproduction}

Five outer folds, stratified by intended assignment, are shared across models and outcomes. Baseline variables are age, gender, state, recruitment cohort, enrollment day, new-member status and the four preceding weekly counts. The $1{,}726$ members with all four counts and those with incomplete baseline histories remain in the same analysis. Training medians impute numeric values, with missingness indicators retained. Natural splines with three degrees of freedom represent age, enrollment day and mean baseline attendance.

The primary learner uses ridge logistic regression. Response, later-action and outcome models include $Z_1$, $(Z_1,S)$ and $(Z_1,S,Z_2)$, respectively. History products and their interactions with baseline attendance and new-member status permit effect modification. Penalties minimize five-fold training cross-validated deviance over $60$ values with a minimum-to-maximum ratio of $10^{-4}$. Inner folds are stratified by each model's response. Preprocessing is learned on the outer training sample and reused during inner tuning. Additive and additive/interaction mixture fits provide comparisons.

All reported calculations use raw predictions. The primary margin is $\epsilon=0.005$. The $0.01$ margin is incompatible at one evaluation row, where a fitted later-action probability reaches $0.9904571$. The weaker margin uses the same cohort and fitted law, without probability clipping. The change followed inspection of the fit. The scientific margin restricts latent assignment probabilities; fitted compatibility does not establish that population restriction. Population endpoints average over all retained members. Principal numerators and masses are averaged separately before division.

\subsection{Sampling inference and comparisons}\label{supp:emp-comparators}\label{supp:emp-inference}

Table~\ref{tab:gym-spi} reports the plug-in and augmented SPI estimates for the population and always-attenders, $U=(1,1)$. The plug-in column gives the reference values in Figure~\ref{fig:gym-gamma}, while the standard errors and confidence intervals refer to the augmented estimates. The SPI estimands, estimates and intervals do not depend on the selected positivity margin.

\begin{table}[t]
\centering
\caption{SPI estimates of the effect of week-2 attendance on week-5 attendance under initial program assignment, at $\theta=1$. All values are percentage points. Standard errors (SE) and nominal pointwise $95\%$ confidence intervals (CI) refer to the augmented estimates and assume independent observations.}\label{tab:gym-spi}
\begin{tabular}{lrrrr}
\toprule
Target & Plug-in & Augmented & SE & $95\%$ CI\\\midrule
Population & 18.62 & 21.09 & 4.74 & $[11.81,30.38]$\\
Always-attenders & 23.24 & 26.04 & 7.44 & $[11.45,40.63]$\\
\bottomrule
\end{tabular}
\end{table}

Augmented SPI intervals use Theorem~\ref{thm:dml} and Corollary~\ref{cor:seq-dml} under their stated conditions. Restricted endpoints use Corollary~\ref{cor:restricted-inference}, including moving candidate derivatives. Table~\ref{tab:gym-restricted-inference} applies the nominal twelve-endpoint adjustment across both targets and $\gamma\in\{0,0.05,0.10\}$. Under the corollary's conditions, this adjustment is conservative for outward enclosure because it uses two-sided endpoint intervals. Its coverage result requires sensitivity inputs specified independently of fitted compatibility checks. Because $\epsilon$ was selected after inspecting the fitted law, that result does not establish nominal coverage for the restricted-endpoint intervals reported here.

\begin{table}[t]
\centering
\caption{One-step endpoint estimates, standard errors and nominal simultaneous $95\%$ outward regions. All entries, including $\gamma$, are percentage points, with $\theta=1$ and $\epsilon=0.005$. The twelve-endpoint formula uses critical value $z_{1-0.05/24}$. Nominal coverage after selecting $\epsilon$ from the data is not established (Section~\ref{supp:emp-inference}).}\label{tab:gym-restricted-inference}
\begin{tabular}{lrrrr}
\toprule
Target & $\gamma$ & Lower (SE) & Upper (SE) & Outer enclosure\\\midrule
Population & 0 & -16.12 (2.63) & 39.62 (4.44) & $[-23.66,52.33]$\\
Population & 5 & -16.71 (2.58) & 42.08 (4.42) & $[-24.11,54.75]$\\
Population & 10 & -17.66 (2.45) & 45.21 (4.39) & $[-24.68,57.78]$\\
Principal & 0 & -16.23 (3.76) & 44.87 (7.08) & $[-27.00,65.15]$\\
Principal & 5 & -16.47 (3.71) & 48.35 (7.04) & $[-27.10,68.53]$\\
Principal & 10 & -17.07 (3.48) & 52.28 (7.00) & $[-27.04,72.35]$\\
\bottomrule
\end{tabular}
\end{table}

Every displayed outer enclosure contains zero. Across all $15$ model/partition fits, all $180$ restricted endpoints are numerically available. For the five interaction fits, the $\gamma=0.10$ upper estimates range from $45.21$ to $46.30$ percentage points for the population and $52.01$ to $55.58$ for the principal group. The zero-budget principal lower estimate varies by $6.90$ points across these interaction fits. The ten largest contributions account for $48.5\%$ and $68.4\%$ of the primary SPI population and principal influence variances. These diagnostics describe variation across fits and influence concentration, without establishing interval calibration.

The richer-history comparator replaces the binary first-week response in the second-stage regressions by its exact released count $L$, integrating over its natural distribution under $Z_1=1$. It uses the same binary attendance intervention and requires consistency, retained-population first-stage exchangeability, positivity and sequential mean exchangeability given $(\X,L)$. The population estimate is $19.05$ percentage points with standard error $5.30$, and the paired richer-minus-coarse estimate is $-2.04$ with standard error $2.23$. The released count includes eleven values of eight, retained as supplied. Sparse high-count categories produce a maximum fitted later-action probability of $0.9987$. This comparison uses different identifying assumptions and cannot validate LSI.

\subsection{Joint coupling and effect-budget sensitivity}\label{supp:emp-joint-sensitivity}

Table~\ref{tab:gym-joint-sensitivity} extracts $\theta\in\{1/4,1,4\}$ and $\gamma\in\{0.10,2\}$ from a grid of $81$ common odds ratios in $[1/16,16]$. For binary outcomes, $\gamma=2$ imposes no effect-difference restriction. The baseline integration measure and raw fitted conditional probabilities remain unchanged. Each principal mass is $\Pn e_{11}(\X;\theta)$, and its endpoint is the mean fitted numerator divided by that mass. The mass does not depend on $\gamma$. The grid describes identification sensitivity without sampling intervals.

\begin{table}[t]
\centering
\small
\caption{Joint sensitivity of fitted-law identification regions at $\epsilon=0.005$. Bounds are in percentage points. The principal mass is the fitted proportion of always-attenders. Unrestricted permits arbitrary latent effect differences. These are exploratory identification calculations, without sampling intervals.}\label{tab:gym-joint-sensitivity}
\begin{tabular}{llrrr}
\toprule
$\theta$ & Effect budget & Population & Principal & Principal mass (\%)\\\midrule
$1/4$ & 10 pp & $[-15.46,40.95]$ & $[-18.50,48.69]$ & $46.24$\\
$1/4$ & Unrestricted & $[-18.32,59.53]$ & $[-28.80,82.96]$ & $46.24$\\
$1$ & 10 pp & $[-19.01,42.89]$ & $[-19.69,48.76]$ & $49.43$\\
$1$ & Unrestricted & $[-22.42,63.38]$ & $[-28.97,81.88]$ & $49.43$\\
$4$ & 10 pp & $[-19.58,42.87]$ & $[-18.92,47.24]$ & $53.57$\\
$4$ & Unrestricted & $[-23.85,64.79]$ & $[-25.74,78.03]$ & $53.57$\\
\bottomrule
\end{tabular}
\end{table}

The unrestricted population benchmark varies $e_{11}(x)$ over its feasible coupling interval at each covariate value, keeping response margins fixed. Its region is $[-25.67,67.24]$ percentage points, lying within the free-binary-confounder region $[-28.88,70.43]$. The calculation uses all finite knots from the unrestricted binary result in Section~\ref{supp:unspecified-computation}, including interior couplings. Population assignment-endpoint sufficiency at fixed coupling does not restrict the coupling search to the two coupling endpoints. This benchmark concerns population bounds without an effect restriction.

\section{Proofs}\label{supp:proofs}
The proof of Theorem~\ref{thm:contraction} appears in Section~\ref{sec:contraction}. The remaining proofs are collected here, using the notation and assumptions of their corresponding statements.

\subsection{Proof of Lemma~\ref{lem:plackett}}\label{supp:proof-plackett}
\begin{proof}
With $B$ and $\Delta$ as in \eqref{eq:plackett}, the four scores are nonnegative exactly on the feasible coupling interval. In its interior,
$\mathrm{OR}(e)=e(1-p_0-p_1+e)/\{(p_0-e)(p_1-e)\}$ increases continuously from zero to infinity. Its unique solution at a finite positive $\theta$ satisfies
\[
 F(e)=(\theta-1)e^2-Be+\theta p_0p_1=0,\qquad F_e=-\sqrt\Delta<0.
\]
This chooses the root in \eqref{eq:plackett} and gives $\partial e/\partial p_z=-F_{p_z}/F_e$. The affine score relations and the piecewise affine formulas at $\theta=0$ and $\infty$ give the remaining derivatives.
\end{proof}

\subsection{Proof of Theorem~\ref{thm:unified-contrast}}\label{supp:proof-unified-contrast}
\begin{proof}
Work conditionally on $\X$ in a cell with $0<w<1$. Assignment compatibility gives
$wh+(1-w)h_0=q$, hence \eqref{eq:unified-h-interval}. LSI gives the full mixture equation $F_a=t_aF_{a1}+(1-t_a)F_{a0}$. Equivalently, the subprobability measure $t_aF_{a1}$ has mass $t_a$ and is dominated by $F_a$. Its smallest and largest first moments select the lower and upper $t_a$ fractions of $F_a$, with fractional allocation of boundary atoms. Therefore
$L_{F_a}(t_a)/t_a\le m_a(1)\le U_{F_a}(t_a)/t_a$.
Convex interpolation between the extremizing submeasures attains every intervening mean; the remaining component is $(F_a-t_aF_{a1})/(1-t_a)$. Substitution of $m_a(1)=\mu_a+(1-t_a)\delta_a$ gives exactly \eqref{eq:unified-delta-interval}. Averaging the two component means gives \eqref{eq:unified-means}. A pure cell has no latent mixture and retains its observed conditional law.

These cellwise conditions are jointly sufficient. Given $\X$, draw $U$ from the fixed response coupling and, conditional on $(\X,U)$, draw each potential outcome independently from its constructed component law. Draw $Z_1$ independently of this vector given $\X$ using its observed propensity. Set $S=U_{Z_1}$, draw $Z_2$ using the chosen cell-specific latent propensity, and set $Y=Y(Z_1,Z_2)$. This reproduces every observed conditional distribution and satisfies both stages of LSI. In particular, separate potential outcomes may have independently chosen component laws, while both second-stage arms share the same assignment probability in each cell.

The cell contribution to $T_c$ is
\[
\sum_a\{wc_{a1}m_a(1)+(1-w)c_{a0}m_a(0)\}
=\sum_a\{A_a\mu_a+B_a(h)\delta_a\}.
\]
For fixed $h$, linearity in each $\delta_a$ gives its endpoint optimization in \eqref{eq:unified-contrast-bounds}. The feasible $h$ interval is compact, and positive $\epsilon$ places its tilted shares inside $(0,1)$. Integrated quantiles are continuous, so both extrema are attained. Regular conditional outcome kernels exist on the stated standard Borel space. A jointly measurable conditional quantile $Q_a(\X,r)$ gives the lower allocated submeasure
\[
A^-_{a,t}(B)=\int_0^t\ind{Q_a(\X,r)\in B}\,dr,
\]
with the upper allocation integrating over $[1-t,1]$. These kernels have mass $t$, are dominated by $F_a$, and fractionally allocate atoms. Interpolating them gives every feasible component mean, using interpolation weight zero when their first moments coincide. The residual submeasure gives the other component law. The feasible graph in $(h,\delta_0,\delta_1)$ is measurable with nonempty compact sections, and the objective is continuous in these variables. A measurable maximizing or minimizing selection therefore exists, including after intersection with a local closed restriction having measurable graph and containing zero. Bounded means and coefficients supply an integrable envelope even when component weights approach zero. The selected kernels have a common full-data completion as above, so the integrals of the conditional extrema are themselves attained.

For completeness, in the unrestricted feasible set of Theorem~\ref{thm:unified-contrast}, scalar values between the extrema are also attained. Write each feasible contrast as
$\delta_a=\ell_{F_a}(t_a)+r_a\{u_{F_a}(t_a)-\ell_{F_a}(t_a)\}$ with $r_a\in[0,1]$. Interpolate the lower and upper optimizers linearly in $(h,r_0,r_1)$, using $r_a=0$ when its interval is a singleton. This gives continuous conditional means along a common interpolation parameter. Bounded outcomes and coefficients justify dominated convergence after integration. The resulting scalar functional is continuous and connects its lower and upper extrema, so the intermediate value theorem supplies every intervening value. This argument concerns scalar identified sets and does not assert convexity of the joint vector of causal means. Under the effect-budget restriction, interval attainment follows instead from the scaling argument in Section~\ref{supp:proof-effect-restrictions}.
\end{proof}

\subsection{Proof of Proposition~\ref{prop:shared-tilt}}\label{supp:proof-shared-tilt}
\begin{proof}
The contraction in Theorem~\ref{thm:contraction} applies simultaneously to both outcome arms, so their population cell laws can be preserved at a common assignment endpoint. At each such endpoint, the latent component laws for the two outcome arms can be optimized separately, giving $J^+_\epsilon$ and $J^-_\epsilon$ in \eqref{eq:shared-tilt-formula}. Averaging their conditional contrast limits gives \eqref{eq:shared-tilt-region}. Sequences with different initial arms allow separate assignments. Theorem~\ref{thm:unified-contrast} supplies measurable full-law attainment in each case.

For a single sequence mean, the same contraction leaves only the two endpoints of its feasible component-share interval. Optimizing the target-arm component laws at these endpoints gives $D^\pm_{F,\epsilon}$, and averaging gives \eqref{eq:full-law-region}. Theorem~\ref{thm:unified-contrast} supplies full-law attainment and every intermediate scalar value.
\end{proof}

\subsection{Proof of Corollary~\ref{cor:binary-principal-bounds}}\label{supp:proof-binary-principal}
\begin{proof}
In a binary mixture, $tm+(1-t)m'=\mu$ with $m,m'\in[0,1]$ gives the stated component limits. At fixed $h$, the target component has proportion $t_a(h)$ for $v=1$ and $1-t_a(h)$ for $v=0$. Both potential-outcome component laws may be chosen separately under the same assignment, giving the displayed conditional contrast limits.

The feasible assignment interval is partitioned at all points satisfying $t_a(h)=\mu_a$ or $t_a(h)=1-\mu_a$. Complementing the mixing proportion leaves this set of points unchanged. On each resulting interval, an upper-contrast branch is a constant plus nonnegative multiples of $1/t_a(h)$ or $1/\{1-t_a(h)\}$. These denominators are positive and affine in mixed cells, so each upper-contrast branch is convex in $h$. Each lower-contrast branch is correspondingly concave. Its required maximum or minimum is therefore attained at a partition endpoint. This proves the finite candidate rule, including coincidences and zero coefficients. Integration with principal-score weights follows from Theorem~\ref{thm:unified-contrast} with the positive, fixed-coupling normalization $\E(e_u)$.

In Example~\ref{ex:principal-endpoint-failure}, the upper objective is $1/\{2(1-h)\}$ for $h\le1/2$ and $1/(2h)$ for $h\ge1/2$. Its maximum is one at $h=1/2$. Reflection gives the lower bound. The stated deterministic potential outcomes reproduce both observed Bernoulli laws under a latent Bernoulli response of probability one half and independent second treatment of probability one half. Other response cells and first-period arms can be completed using the theorem's construction, so the example embeds in a full longitudinal law.
\end{proof}

\subsection{Proof of Proposition~\ref{prop:finite-effect-budget}}\label{supp:proof-finite-effect-budget}
\begin{proof}
For binary outcomes,
\[
\ell_a(t)=\max\left\{-\frac{\mu_a}{1-t},-\frac{1-\mu_a}{t}\right\},
\qquad
u_a(t)=\min\left\{\frac{1-\mu_a}{1-t},\frac{\mu_a}{t}\right\}.
\]
Partition the feasible $h$ interval at the feasible solutions of
$t_a(h)\in\{\mu_a,1-\mu_a\}$. On each resulting interval every rectangle
edge is a constant divided by a positive affine function of $h$. The at most twelve vertices are those described in Corollary~\ref{cor:effect-heterogeneity-bounds}. For each formula, retain the interval endpoints,
all roots at which one of its feasibility inequalities becomes an equality,
and all stationary points of its scalar target. Evaluate only feasible
candidates. An inequality or derivative that vanishes identically contributes
no roots. Repeated roots are retained once, and a singleton assignment interval
requires only its sole value. A vertex feasible only at an isolated point is retained by its
feasibility equalities; a constant objective on a feasible interval is
attained at its boundary. Thus the list also covers degeneracies.

All required nontrivial equations have degree at most two. A feasibility
comparison subtracts at most two reciprocal affine functions and a constant.
Its denominator is positive, so clearing it gives a quadratic. A corner
objective is a constant plus at most two reciprocal affine functions, whose
derivative has a quadratic numerator. A strip-edge objective is a constant
plus one reciprocal affine function and a linear term; its derivative also
has a quadratic numerator. Every smooth interior extremum is stationary, and
every other extremum lies at one of the retained boundaries. This proves the
finite algorithm for population or principal linear targets. Principal
normalization occurs after integrating the cell contributions. At
$\gamma\ge2$, Corollary~\ref{cor:binary-principal-bounds} supplies the simpler
unrestricted candidate list.
\end{proof}

\subsection{Proof of Lemma~\ref{lem:mixture}}\label{supp:proof-mixture}
\begin{proof}
Consistency and initial ignorability give the first posterior. Total probability and Bayes' rule give the next two identities, and LSI permits averaging the unselected component outcomes to obtain the last.
\end{proof}

\subsection{Proof of Proposition~\ref{prop:bias}}\label{supp:proof-bias}
\begin{proof}
The mixture identity in Lemma~\ref{lem:mixture} gives
\[
 \mu_{\zbar}-\sum_u w_um_{\zbar}
 =\frac{\Cov_w(m_{\zbar},h^{(z_2)})}{q_{z_2}}
 =\frac{(-1)^{1-z_2}w(1-w)\delta_{\zbar}\rho}{q_{z_2}}.
\]
Integrating gives (i). Solving $wh_1+(1-w)h_0=q_1$ and $h_1-h_0=\rho$ gives (ii). Integrability and positivity justify the expectations.
\end{proof}

\subsection{Proof of Proposition~\ref{prop:full-law-compatibility}}\label{supp:proof-full-law-compatibility}
\begin{proof}
The allocation argument of Theorem~\ref{thm:unified-contrast} gives \eqref{eq:outcome-compatibility}. For binary outcomes, $m_1=\mu+(1-t)\delta$ and $m_0=\mu-t\delta$ in $[0,1]$ give \eqref{eq:binary-delta-compatibility}. The construction extends to integrable outcomes since uniform positivity gives $\E|Y(j,a)|\le\epsilon^{-2}\E|Y|$.
\end{proof}

\subsection{Proof of Proposition~\ref{prop:bounds}}\label{supp:proof-support-bounds}
\begin{proof}
The notation $f,g,t^\pm_\epsilon,I_\epsilon$ is that of Section~\ref{supp:sens-est}. At a positive margin, LSI transfers observed support to every positive-score component. Maximizing $(w-t)(m_1-m_0)$ gives the decreasing $f(t)$ for $t\le w$ and increasing $g(t)$ for $t\ge w$, hence the endpoint rule. Theorem~\ref{thm:unified-contrast} makes these support constraints exact for endpoint-supported laws and outer bounds otherwise.

Increasing $\epsilon$ raises $t^-_\epsilon$ and lowers $t^+_\epsilon$. Under nondegeneracy, strict monotonicity of $f,g$ and dominated convergence give $I_\epsilon\to I_0$. Since $t^-<w<t^+$ in mixed cells, each positive margin excludes the limiting extrema, strictly shrinking the integrated interval when such cells have positive probability. The nested union over $0<\epsilon\le\epsilon^*$ gives (ii), since any law with a larger margin also has margin $\epsilon^*$. Nondegeneracy is essential. At $\mu=y_{\min}$, every interior tilt forces both means to $y_{\min}$, whereas zero selected mass can hide a different mean; the analogous issue holds at $y_{\max}$. Equation~\eqref{eq:bias} and construction L5 in Section~\ref{app:relations} prove (iii).
\end{proof}

\subsection{Proof of Corollary~\ref{cor:effect-heterogeneity-bounds}}\label{supp:proof-effect-restrictions}
\begin{proof}
Theorem~\ref{thm:unified-contrast} characterizes all compatible component
means. Intersecting its two mean intervals with the stated restriction is
therefore necessary and sufficient. The intersection is a compact polygon
containing zero, and its linear extrema occur at its listed vertices, also
when it degenerates to a segment or point. Conditional component laws
attaining these means, and their measurable full-law construction, are
provided by that theorem. Scaling an optimizing $\bm d$ towards zero at its
fixed $h$ preserves the polygon and connects its target value to the common
observed-law value. Applying this to the lower and upper optimizers proves
attainment of every scalar intermediate value after averaging. The effect
restriction is preserved by common contractions, so
Theorem~\ref{thm:contraction} gives the population endpoint
rule. The remaining claims follow directly from the rectangle's largest
absolute coordinate difference and $|\delta_a|\le y_+-y_-$.
\end{proof}

\subsection{Proof of the finite search over unspecified couplings}\label{supp:proof-unspecified-coupling}
\begin{proof}
The notation and candidate partition are those of Section~\ref{supp:unspecified-computation}. Each directional mean-deviation envelope, between its Bernoulli trimming knots, is constant or has form
\[
\frac{K u}{1-cu},\qquad K\ge0,\qquad
\frac{d^2}{du^2}\frac{K u}{1-cu}=
\frac{2Kc}{(1-cu)^3}\ge0.
\]
Inside each partition interval the feasible assignment-boundary labels and active envelope branches are fixed. Each boundary's summed deviation is convex in $r$, as are its maximum over boundaries and the response-weighted sum. Its maximum is therefore at a partition endpoint. Pure-cell values follow by continuity and equal the observed contrast.

Choosing the least extremising candidate is measurable. The other initial arm can be completed with identical component outcome laws and observed assignment propensities. Theorem~\ref{thm:unified-contrast} therefore supplies full-law attainment over all measurable couplings. Each fixed-coupling interval contains the same population SPI functional, so their union is an interval.
\end{proof}

\subsection{Proof of Corollary~\ref{thm:union}}\label{supp:proof-population-identification}
\begin{proof}
Under \eqref{eq:union}, every mixed-cell term in the bias formula of Proposition~\ref{prop:bias} is zero, and pure cells also have zero bias. Hence $\psi(j,a)=\psi_g(j,a)$. SPI sets $\delta_{(j,a)}=0$, while LPH sets $\rho=0$. The argument holds for every response coupling compatible with the observed margins, so no coupling specification is needed.
\end{proof}

\subsection{Proof of Corollary~\ref{thm:pce}}\label{supp:proof-principal-identification}
\begin{proof}
Under SPI, the components in \eqref{eq:tilted-mixture} coincide. Hence $m_{\zbar}(u,\X)=\mu_{\zbar}(u_{z_1},\X)$ on positive-score strata, and averaging over the marginal law of $\X$ with weights $e_u(\X)/\E(e_u)$ gives \eqref{eq:pce-id}. Contrasts follow by subtraction.
\end{proof}

\subsection{Proof of Proposition~\ref{prop:relations}}\label{supp:proof-relations}
\begin{proof}
Under LPH the latent propensity is constant in each response cell. Summing the LSI factorization over $U$ then factors the joint law of $Z_2$ and the entire outcome vector, giving SI. Equation~\eqref{eq:mean-si} and Proposition~\ref{prop:bias} prove the mean and bias claims. The constructions in Section~\ref{app:relations} and Remark~\ref{rem:nonnested} verify the remaining claims explicitly.
\end{proof}

\subsection{Proof of Proposition~\ref{prop:residual-selection}}\label{supp:proof-residual-selection}
\begin{proof}
Consistency now gives $\mu_{\zbar}=\sum_u\widetilde w_u^{z_2}(m_{\zbar}+r_{\zbar})$. Subtract the untilted mean and apply the covariance identity in Proposition~\ref{prop:bias}. The residual weights sum to one, giving the bounds. A stratum-specific allowance has sequence radius
\[
 \E\sum_s p_{z_1}^s\sum_u\widetilde w_u^{z_2}\eta_{\zbar}(u,\X).
\]
\end{proof}

\subsection{Proof of Proposition~\ref{prop:residual-compatibility}}\label{supp:proof-residual-compatibility}
\begin{proof}
Let $n_v$ be the potential-outcome mean among those receiving the target second treatment and $a_v^{\rm mis}$ its mean among those receiving the other treatment, conditional on $(\X,Z_1=z_1,U=u^{(v)})$. The identities
\[
\mu=tn_1+(1-t)n_0,\qquad n_1-n_0=\delta_{\zbar}+r_1-r_0,
\qquad m_v=k_vn_v+(1-k_v)a_v^{\rm mis}=n_v-r_v
\]
give \eqref{eq:residual-components}. Assignment compatibility gives selected shares $t$ and $1-t$. The selected component distributions must mix to $F$, and Proposition~\ref{prop:full-law-compatibility} characterizes precisely their attainable means by the trimmed-mean condition. Since $F$ is supported on $[a,b]$, its component distributions have that support. The missing distributions require their means in $[a,b]$, proving necessity.

For sufficiency, choose selected component distributions with means $n_v$ that mix to $F$, using the allocation construction of Proposition~\ref{prop:full-law-compatibility}, and choose missing distributions on $[a,b]$ with means $a_v^{\rm mis}$. Draw $(\X,U)$ with the specified principal scores. Given $(\X,U)$, draw indicators $A_0,A_1$ with $\Prob(A_j=1\mid\X,U)=h_U(j,\X)$. For every $(j,a)$, draw $Y(j,a)$ from its selected component law when $A_j=a$ and from its missing law otherwise. The same $A_j$ is used for both potential outcomes in arm $j$. Finally draw $Z_1$ independently of this vector given $\X$ with its observed propensity, set $Z_2=A_{Z_1}$, $S=U_{Z_1}$ and $Y=Y(Z_1,Z_2)$. This joint construction reproduces every observed conditional outcome law and all specified means and residuals, and preserves first-stage ignorability and positivity. Point masses at the missing means suffice without support restrictions. Uniform positivity, bounded residuals and an integrable observed outcome ensure integrability of that construction.
\end{proof}

\subsection{Proof of the common cross-fitting expansion}\label{supp:common-crossfit}
\begin{proof}
For independent identically distributed observations, use fixed $K$ folds with $n_k/n$ bounded away from zero, fit outside each evaluation fold and average contributions with weights $n_k/n$. This follows the cross-fitting construction of \citet{chernozhukov2018}. Suppose, for every required $G$, its foldwise fitted contribution is $L_2(\Prob)$ consistent and
\[
\E(\widehat{\mathcal C}_G\mid\text{training fold})-\E G=o_{\Prob}(n^{-1/2}).
\]
For folds assigned independently of the observations, conditional Chebyshev gives
$\sqrt{n_k}(\Pn^k-\E)(\widehat{\mathcal C}_G-\mathcal C_G)=o_{\Prob}(1)$.
The same conclusion holds for the stratified scheme in Section~\ref{supp:distribution-scores}. Write $\omega_z=\Prob(Z_1=z)>0$, let $n_{kz}$ count observations with $Z_1=z$ in fold $k$, and put $\Delta_G=\widehat{\mathcal C}_G-\mathcal C_G$. Conditional on all initial-treatment labels and the training observations, uniform splitting within each stratum leaves the evaluation observations independent with laws $\Prob(\cdot\mid Z_1=z)$. Conditional Chebyshev therefore applies after centering within strata. Since $n_{kz}/n_k-\omega_z=O_{\Prob}(n^{-1/2})$, the difference from centering under $\Prob$ is
\[
\sum_z\left(\frac{n_{kz}}{n_k}-\omega_z\right)
\E(\Delta_G\mid Z_1=z,\text{training})
=o_{\Prob}(n^{-1/2}),
\]
where $L_2(\Prob)$ consistency makes each conditional mean $o_{\Prob}(1)$. These expectations integrate a new observation while holding the fitted function fixed.
Summing the fixed folds yields
$\Pn\widehat{\mathcal C}_G-\E G=(\Pn-\E)\mathcal C_G+o_{\Prob}(n^{-1/2})$.
Applying this jointly to numerators and masses and expanding at $D>0$ proves the ratio expansion in \eqref{eq:common-ratio}. The multivariate central limit theorem gives joint normality. Conditional $L_2$ consistency also gives $\Pn(\widehat\phi-\phi)^2=o_{\Prob}(1)$, so empirical influence covariances consistently estimate their population covariance. Estimator covariance is this matrix divided by $n$. No Donsker condition is needed for this route. The inference results in Section~\ref{supp:inference-statements} supply their own drift, score-continuity and domain conditions.
\end{proof}

\subsection{Proof of Theorem~\ref{thm:dml}}\label{supp:proof-eif-pce}\label{supp:proof-dml}
\begin{proof}
With $s^*=u_j$ and $\mu=\mu_{(j,a)}(s^*,\X)$, the choices $A=e_u\mu$ and $m=e_u$ in \eqref{eq:common-contribution} give $\mathcal C_A=\mu W_u+e_uB_{(j,a),s^*}$ and $\mathcal C_m=W_u$, where $W_u$ is defined in Section~\ref{app:eif-pce}. Equation~\eqref{eq:common-ratio} proves the formula. Positivity and bounded score derivatives give square integrability. SPI imposes no additional observed-law restriction on the stipulated support, so bounded observed-law scores generate a dense tangent space and the gradient is efficient.

For the estimation statement, assume (R1)--(R2) and retain the notation of Section~\ref{app:mr}, with $\zbar=(j,a)$ and $\psi_u=\psi_u(j,a)$. The positive mass and consistent nuisance fits keep the estimated denominator away from zero with probability tending to one. The fitted residual $B'$ is $B_{(j,a),s^*}$ evaluated at $\Gamma'$.

Use primes for $\Gamma'=(\pi',p',q',\mu')$, with $\mu'=\mu'_{(j,a)}(s^*,\X)$, and define
\begin{align*}
r'_u=r_u(\X;p')&=e_u(p)-e_u(p')-\sum_z\dot e_{u,z}(p')(p_z-p'_z),& |r'_u|&\le C|p-p'|^2,\\
c'&=\sum_z\dot e_{u,z}(p')(\pi_z/\pi'_z-1)(p_z-p'_z),&
R'&=\frac{\pi_jp_j^{s^*}q_a(j,s^*,\X)}{\pi'_jp_j'^{\,s^*}q'_a(j,s^*,\X)}.
\end{align*}
The first bound is Taylor's theorem on $\calK$. With
$\phi(\calO;\Gamma',t)=(\mu'-t)W'_u+e'_u B'$,
conditional expectations give the exact identities
\begin{equation}\label{eq:phi-affine}
\E\phi(\calO;\Gamma',t)
=b(\Gamma')-\{t-\psi_u(\zbar)\}D(\Gamma'),\qquad
D(\Gamma')=\E(e_u-r'_u+c'),
\end{equation}
\begin{equation}\label{eq:exact-bias}
b(\Gamma')=-\E\{(\mu'-\psi_u)r'_u\}
+\E\{(\mu'-\psi_u)c'\}
+\E\{(\mu'-\mu)(e_u-e'_uR')\}.
\end{equation}
Indeed $\E(A'_z\mid\X)=(\pi_z/\pi'_z)(p_z-p'_z)$ and
$\E(B'\mid\X)=R'(\mu-\mu')$.
Write $r_p=\|p'-p\|_2$, and likewise $r_\pi,r_q,r_\mu$. Since
$|c'|\le C|\pi-\pi'||p-p'|$ and
$|e_u-e'_uR'|\le C(|p-p'|+|\pi-\pi'|+|q-q'|)$,
\begin{equation}\label{eq:bias-bound}
|b(\Gamma')|\le C\{r_p(r_p+r_\pi)+r_\mu(r_\pi+r_p+r_q)\}.
\end{equation}

The exact drift and \eqref{eq:dml-rates} give orthogonality and negligible bias. For the fitted score difference $\Delta=\phi(\calO;\widehat\Gamma,\psi_u)-\phi(\calO;\Gamma,\psi_u)$, bounded nuisance factors and (R2) give
\[
\E(\Delta^2\mid\text{training})\le C\|\widehat\Gamma-\Gamma\|_4^2=o_{\Prob}(1),
\]
using $\|h\|_4^2\le\|h\|_\infty\|h\|_2$ and bounded suprema. Uniform convergence is not required here. The common cross-fitting argument in Section~\ref{supp:common-crossfit} proves the result. All nuisance rates $o_{\Prob}(n^{-1/4})$ suffice, but \eqref{eq:dml-rates} retains the sharper product structure.
\end{proof}

\subsection{Proof of Corollary~\ref{cor:seq-dml}}\label{supp:proof-sequence-inference}
\begin{proof}
For $\zbar=(j,a)$, write $\mu_s=\mu_{(j,a)}(s,\X)$ and use primes for fitted quantities. The fitted centered score below is evaluated at the true target $\psi_g(\zbar)$.
For $R'_s=\pi_jq_a(j,s,\X)/(\pi'_jq'_a(j,s,\X))$, the exact sequence drift is
\begin{equation}\label{eq:sequence-drift}
\E\varphi^{g\prime}_{\zbar}
=\E\sum_s\left\{(p_j'^{\,s}-p_j^s)\mu'_s
\frac{\pi'_j-\pi_j}{\pi'_j}
+p_j^s(\mu'_s-\mu_s)(1-R'_s)\right\}.
\end{equation}
Cauchy--Schwarz bounds it by the products in the rate condition of Corollary~\ref{cor:seq-dml}. Score continuity and the argument in Section~\ref{supp:common-crossfit} apply with unit mass.
\end{proof}

\subsection{Proof of Theorem~\ref{thm:distribution-inference}}\label{supp:proof-distribution-inference}
\begin{proof}
The notation and conditions are those of Section~\ref{supp:distribution-regularity}. For a conditional law put $H^-(c,t)=tc-\E_F(c-Y)_+$ and
$H^+(c,t)=tc+\E_F(Y-c)_+$. Their respective maximum and minimum occur at the lower and upper cutoffs, giving
\begin{equation}\label{eq:tail-hinges}
L_F(t)=H^-(Q^-,t),\qquad U_F(t)=H^+(Q^+,t).
\end{equation}
At a continuous cutoff $\partial_cH^\sigma=0$; a cutoff inside a fixed atom is locally unchanged. Differentiating the law and $t$ therefore gives the two terms of \eqref{eq:tail-gradients}.

For continuous cutoffs the density bounds yield the quantitative estimates
\begin{equation}\label{eq:cutoff-remainders}
\begin{aligned}
|H^\sigma(c,t)-R_F^\sigma(t)|&\le\tfrac12\overline f|c-Q_F^\sigma(t)|^2,\\
|Q_F^\sigma(t')-Q_F^\sigma(t)|&\le |t'-t|/\underline f,\\
|R_F^\sigma(t')-R_F^\sigma(t)-Q_F^\sigma(t)(t'-t)|&\le |t'-t|^2/(2\underline f).
\end{aligned}
\end{equation}
The first inequality integrates the difference between the conditional distribution function and its target probability between $c$ and $Q$; the others integrate the inverse-density bound. The last inequality also holds with $Q_F^\sigma(t')$ in the linear term. For a preserved cutoff atom, $R_F^\sigma$ is affine in $t$ and all these remainders vanish. Uniform consistency retains the specified neighbourhoods or atoms.

For a bounded centred score $\ell$, the path $d\Prob_\eta=(1+\eta\ell)d\Prob$ changes conditional probabilities and bounded moments uniformly by $O(|\eta|)$. Strict observed positivity preserves compatibility at the latent margin. Separation preserves assignment and optimizer branches and, on fixed finite support, the strict trimming gap without moving atom locations. These paths supply every bounded score in the observed-law model on its stipulated support. The derivative of $R_r$ is $Q_r\partial t_r$ plus the conditional-law derivative represented by $W_r\xi_r$. Applying \eqref{eq:score-conv} to $\E K$ gives the complete derivatives \eqref{eq:distribution-total-derivative} and gradient \eqref{eq:distribution-eif}. Bounded derivatives justify differentiation under expectation. Density of bounded scores in $L_0^2(\Prob)$ proves efficiency.

For one fitted fold, work on the event of preserved branches and cutoffs and put
$\Delta\nu=\widehat\nu-\nu$, $\Delta R_r=\widehat R_r-R_r$ and $\Delta Q_r=\widehat Q_r-Q_r$.
Base residuals centred at fitted coordinates but using true denominators have conditional means $-\Delta\nu_l$. Equations~\eqref{eq:tail-hinges}--\eqref{eq:cutoff-remainders} similarly give
\begin{equation}\label{eq:tail-drift}
\E(\widehat\xi_r\mid\X,j_r,s_r,a_r)
=-\Delta R_r+\widehat Q_r(\widehat t_r-t_r)
+O(|\Delta Q_r|^2+|\Delta\nu|^2).
\end{equation}
One obtains this bound by expanding at the true-law cutoff for $\widehat t_r$ and then between the true and fitted tilts. The inverse-density bound controls their cross-products. Backward Taylor expansion of $K$ at its fitted arguments cancels the linear base and tail errors against these residual means. Its remaining tilt term cancels the second term of \eqref{eq:distribution-total-derivative}, since
$\widehat t_r-t_r=\nabla t_r(\widehat\nu)^\top\Delta\nu+O(|\Delta\nu|^2)$.
The integrated remainder with true denominators is consequently at most
$C\E\{|\Delta\nu|^2+\sum_r(|\Delta R_r|^2+|\Delta Q_r|^2)\}$.

Replacing denominators by fitted ones multiplies residual means by bounded smooth probability ratios. Cauchy--Schwarz bounds the added error by $C(r_n^2+s_nr_n)$. In particular, \eqref{eq:tail-drift} introduces no first-order cutoff error. Thus the fitted contribution has drift $O_{\Prob}(r_n^2+s_nr_n)=o_{\Prob}(n^{-1/2})$. Boundedness and Lipschitz hinge functions give $L_2$ score convergence on the preserved event, and bounded contributions control its vanishing-probability complement. The common cross-fitting argument in Section~\ref{supp:common-crossfit} proves the joint expansion and covariance consistency.
\end{proof}

\subsection{Proof of Corollary~\ref{cor:restricted-inference}}\label{supp:proof-restricted-inference}
\begin{proof}
On the preserved charts of Section~\ref{supp:restricted-inference}, the numerator and mass have bounded Lipschitz gradients. Base residuals centered at fitted coordinates but using true denominators have conditional means $\nu-\widehat\nu$. Taylor expansion about $\widehat\nu$ therefore cancels the linear error in each fitted contribution. Replacing the denominators by fitted values introduces bounded smooth probability ratios, so
\[
\left|\E(\widehat{\mathcal C}_G\mid\text{training})-\E G\right|
\le C(r_n^2+s_nr_n),\qquad G\in\{A_k,m_k\}.
\]
Bounded residuals and Lipschitz gradients give $L_2$ score consistency. Uniform nuisance consistency preserves the charts with probability tending to one. The ratio argument in Section~\ref{supp:common-crossfit} and the union bound then prove the claims.
\end{proof}

\end{document}